\documentclass[sigconf]{acmart}

\AtBeginDocument{%
  \providecommand\BibTeX{{%
    \normalfont B\kern-0.5em{\scshape i\kern-0.25em b}\kern-0.8em\TeX}}}

\setcopyright{acmcopyright}

\copyrightyear{2022}
\acmYear{2022}
\setcopyright{acmcopyright}\acmConference[SIGMOD '22]{Proceedings of the 2022 International Conference on Management of Data}{June 12--17, 2022}{Philadelphia, PA, USA}
\acmBooktitle{Proceedings of the 2022 International Conference on Management of Data (SIGMOD '22), June 12--17, 2022, Philadelphia, PA, USA}
\acmPrice{15.00}
\acmDOI{10.1145/3514221.3526177}
\acmISBN{978-1-4503-9249-5/22/06}

\usepackage{graphicx}
\usepackage[utf8]{inputenc} 
\usepackage{amsmath,amsthm} 
\usepackage{url}
\usepackage{fancyvrb}
\usepackage{wrapfig}
\usepackage{xspace}
\usepackage{subcaption}
\usepackage{booktabs}
\usepackage{hyperref}
\hypersetup{
colorlinks=false,
citebordercolor={green}
}

\newcommand{\fn}[1]{\textcolor{pink}{#1}}
\newcommand{\draco}[1]{\textcolor{black}{#1}}
\newcommand{\jinshu}[1]{\textcolor{black}{#1}}
\newcommand{\eat}[1]{}

\newtheorem{lemma}{Lemma}

\newcommand{\x}{\textbf{x}}
\newcommand{\y}{\textbf{y}}
\newcommand{\X}{\textbf{X}}

\newcommand{\nethist}{{\sc Network-Construct-Histo}\xspace}
\newcommand{\netreal}{{\sc Network-Construct-RealTime}\xspace}
\newcommand{\preproc}{{\sc Preprocessing}\xspace}
\newcommand{\netapprox}{{\sc Network-Approximate}\xspace}
\newcommand{\pruning}{{\sc Pruning}\xspace}
\newcommand{\name}{\texttt{TSUBASA}\xspace}

\usepackage{etoolbox,xspace}
\usepackage{algorithmicx}
\usepackage{algorithm}
\usepackage{algpseudocode}
\algtext*{EndWhile}
\algtext*{EndIf}
\algtext*{EndFunction}
\algtext*{EndFor}
\algtext*{EndElse}
\algrenewcommand\algorithmicrequire{\textbf{Input:}}
\algrenewcommand\algorithmicensure{\textbf{Output:}}

\newcommand{\squishlist}{
 \begin{list}{$\bullet$}
  { \setlength{\itemsep}{0pt}
     \setlength{\parsep}{1pt}
     \setlength{\topsep}{1pt}
     \setlength{\partopsep}{0pt}
     \setlength{\leftmargin}{1em}
     \setlength{\labelwidth}{1em}
     \setlength{\labelsep}{0.5em} } }

\newcommand{\squishend}{\end{list}
}

\begin{document}

\title{TSUBASA: Climate Network Construction\\ on Historical and Real-Time Data}

\author{Yunlong Xu}
\authornote{Both authors contributed equally to this research.}
\orcid{}
\author{Jinshu Liu}
\authornotemark[1]
\affiliation{
  \institution{}
  \streetaddress{}
  \city{University of Rochester}
  \state{}
  \country{}
  \postcode{}
}
\email{yxu103@u.rochester.edu}
\email{jliu158@ur.rochester.edu}

\author{Fatemeh Nargesian}
\affiliation{%
  \institution{University of Rochester}
  \streetaddress{}
  \city{}
  \country{}}
  \email{fnargesian@rochester.edu}

\renewcommand{\shortauthors}{Xu, Liu, and Nargesian}

\renewcommand{\shorttitle}{TSUBASA: Climate Network Construction on Historical and Real-Time Data}

\begin{abstract}
    A climate network represents the global climate system by the interactions of a set of anomaly time-series. Network science has been applied on climate data to study the dynamics of a climate network. 
    The core task and first step to enable interactive network science on climate data is the efficient construction and update of a climate network on user-defined time-windows. We present \name, an algorithm for the efficient  construction of climate networks based on the exact calculation of  Pearson's correlation of large time-series. By pre-computing simple and low-overhead  statistics, \name can efficiently compute  the exact pairwise correlation of time-series on arbitrary time windows at query time. For real-time data, \name proposes a fast and incremental way of updating a network at interactive speed. 
    Our experiments show that \name is faster than  approximate solutions at least one order of magnitude for both historical and real-time data and outperforms a baseline for time-series correlation calculation up to two orders of magnitude. 
\end{abstract}

\eat{
\ccsdesc[500]{Computer systems organization~Embedded systems}
\ccsdesc[300]{Computer systems organization~Redundancy}
\ccsdesc{Computer systems organization~Robotics}
\ccsdesc[100]{Networks~Network reliability}
}

\keywords{time-series, climate data, climate networks, correlation matrix}

\maketitle

\section{Introduction}
\label{sec:intro}

{To identify and analyze patterns in global climate, 
scientists and climate risk analysts model climate data as complex networks -- networks with non-trivial topological properties~\cite{abe2004scale,kim2019complex,gozolchiani2008pattern}.}  
The  climate network architecture  represents the global climate system by a set of 
anomaly time-series (departure from the usual behavior) of gridded climate data and their interactions~\cite{Tsonis2004}. 
A climate data set includes remote and in-situ sensor measurements (e.g. sea surface temperature and sea level pressure) covering a grid (e.g.  with a resolution of $2.5^{\circ}\times 2.5^{\circ}$). 
Nodes in a  climate network are geographical locations, characterized by time-series and edges represent information flow between nodes. The 
edge weights indicate a degree of correlation between the behaviors of time-series (e.g. Pearson's correlation). Note the geographical locality of nodes does not directly imply the topology of a network.  

{Several studies have applied network science on climate data assuming dynamic networks that are changing  with real-time data~\cite{Berezin2012}.   
Climate networks have been shown to be powerful tools for gaining insights on earthquakes~\cite{abe2004scale}, rainfalls~\cite{kim2019complex}, and global climate events such as El Niño~\cite{gozolchiani2008pattern}. } 

{The common way for network dynamics analysis is to construct networks for each hypothesized time-window and analyze them separately~\cite{Faghmous:2014}. 
Figure~\ref{fig:application} shows the steps of constructing a climate network. Given a query window provided by a user, a correlation matrix is constructed by computing the pairwise correlation of all time-series on the query window. Pearson’s correlation is one of the most dominant measures for studying the pairwise climatical correlation~\cite{donges2009complex}. The correlation matrix enables  visualization ~\cite{nocke2015visual}, network dynamics analysis~\cite{Berezin2012}, as well as tasks such as community detection ~\cite{tantet2014interaction}. To analyze the topology of the network, a user-provided correlation threshold can be applied on the matrix to find the significant edges between nodes and obtain a boolean network matrix. 
From the mathematical perspective, the analytical computation of the evolution of a complex system (or even not so complex such as Ordinary Differential Equation systems), depends on the robustness and correctness of the initial weights in the complex network~\cite{gozolchiani2008pattern}. 
}

{The core task in network construction is the problem of  large-scale all-pair time-series correlation calculation. The key challenges of interactive network analysis include: 1) exact calculation of the complete correlation matrix, 2) correlation calculation on time-windows of arbitrary size, and 3) efficiency of network construction and update  for historic and real-time data to achieve interactivity.}

\begin{figure}[!t]
    \centering
    \includegraphics[width=\linewidth]{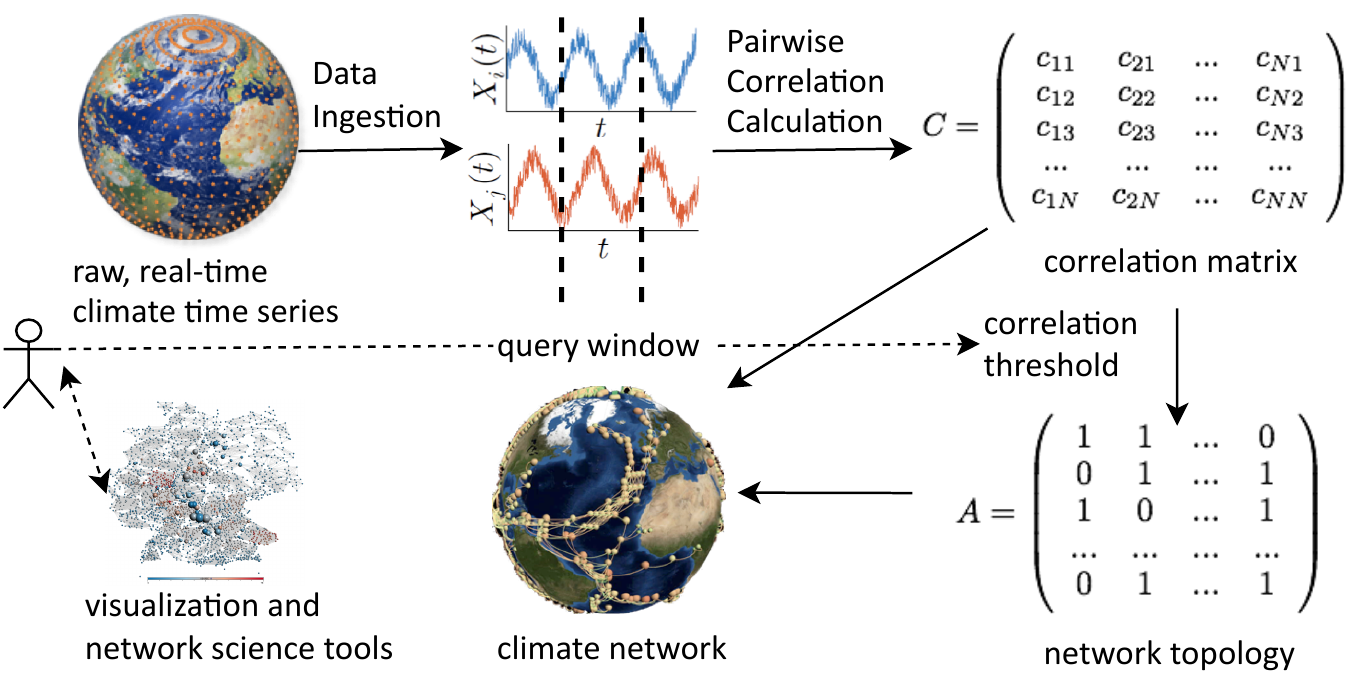}
    \vspace{-5mm}
    \caption{{Climate Network Construction.}}
    \label{fig:application}
\end{figure}

{The line of data management research that computes  networks on time-series (for example, for stock market data or climate data) apply pruning techniques on the approximation of correlation measures. 
In particular, StatStream~\cite{ZhuS02}  and MASS~\cite{MueenNL10} reduce the correlation of time-series to the distance  of their Discrete Fourier Transform (DFT) coefficients and propose grid-based indexing~\cite{ZhuS02} 
and I/O-aware techniques~\cite{MueenNL10} for performing   threshold-based correlated time-series search. 
The accuracy of the network can be increased by considering a very large number of DFT coefficients that are expensive to compute. 
a query window size is restricted to be an integral multiple of the basic window size, which limits the usability of these algorithms. }
{In this paper, we present \name a framework for efficient construction and update  of the correlation matrix  for arbitrary query windows on historical and real-time data. The differences of \name and existing work are three-fold. First, unlike these algorithms that only identify edges  with a correlation higher than a threshold, \name computes and updates the complete  correlation matrix. This enables network dynamics analysis as well as choosing arbitrary thresholds at query time.  Second, \name computes the exact correlation between time-series. Finally, unlike the existing work, \name allows arbitrary query window sizes.    
\eat{We present the mathematics for efficiently computing all-pair correlations on arbitrary query window sizes by relying  on pre-computed sketches collected 
by one pass of data. 
The sketches include statistics that are pre-computed by one pass over data and can be reused at query time for constructing networks for  arbitrary queries.  
Moreover, using these sketches,  
the correlation matrix can be 
updated as new data points arrive real time, in an incremental way, without recomputing the correlation from scratch.} 
} 
In this paper, we make the following contributions. 
{
\squishlist
\item We present the mathematical tools for the exact calculation of pairwise Pearson's correlation of  time-series using the basic window model. 
\item \name relies on simple statistics of basic windows pre-computed by doing one pass over the whole data. This   provides a flexible and highly responsive correlation calculation mechanism. Users can obtain a correlation matrix given any query window  without computing the correlation statistics repeatedly. 
\item We propose an incremental solution for real-time update of the  correlation matrix and climate network.  Relying on the easy-to-compute  statistics of basic windows means the correlations can be updated quickly for frequently-updated time-series.  
\item We enable queries with arbitrary time-window size and start and end point on both historic and real-time by relaxing the  restriction of the existing basic window model on a query window size being an integral multiple of basic window size.
\squishend
}

\begin{figure*}[!ht]
    \centering
    \includegraphics[width=\textwidth]{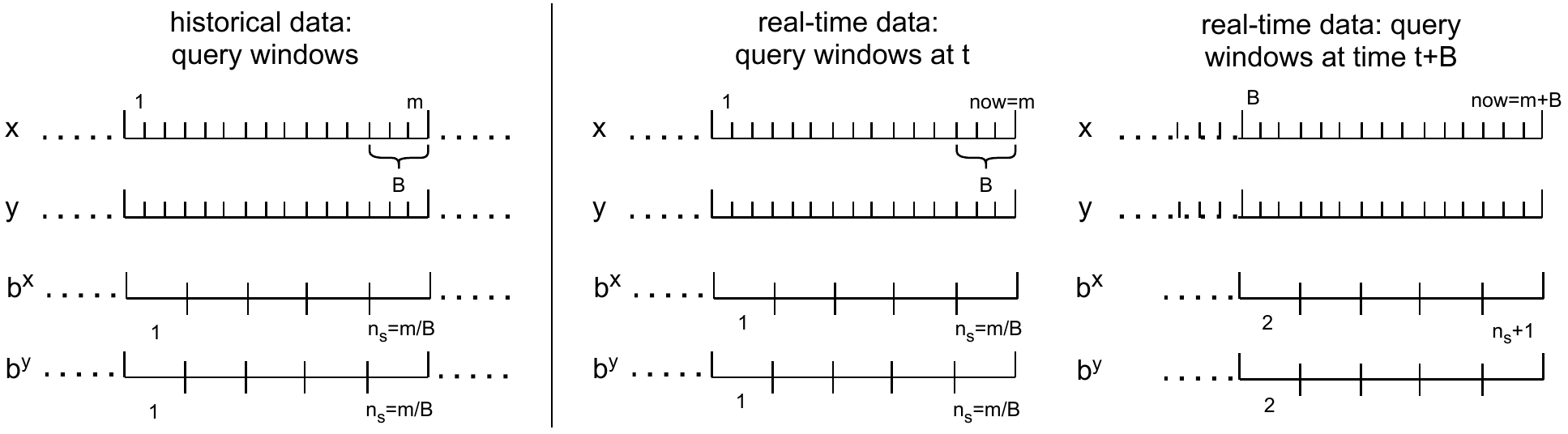}
    \vspace{-3mm}
    \caption{Query Windows for Historical and Real-time Data.}
    \label{fig:windows}
\end{figure*}

\section{Preliminaries} 
\label{sec:prelim}

We start by introducing the climate network  construction problem, then, give an overview of existing approximate techniques for calculating  time-series correlation. 

\subsection{Climate Network} 

We are given a collection $\mathcal{L}=\{x^1, \ldots, x^n\}$ of geo-labeled time-series, where $x^i$ 
denotes the time-stamped values of a climatic variable collected at location $i$. A time-series $x^i$ is defined as 
$[\x^i_1, \ldots, \x^i_m]$, where $\x^i_j$ is the observed value at time 
$j$. 
We assume all  time-series in $\mathcal{L}$ are synchronized, i.e. each time-series has a  value available at every periodic  time interval, namely time resolution. 
Particularly, if the time resolution is $\gamma$ and the current timestamp is $j$, every $x^i$ is $\mathcal{L}$ will have a value observed at time $j+\gamma$. If an $x^i$ has missing value at $j$, 
a value is interpolated or if multiple values appear between $j$ and 
$j+\gamma$, an aggregate value is assigned. 
Table~\ref{tbl:notations} shows a list of notations used throughout this paper. 

Now, we turn our attention to user's  query. 
At query time, a user defines a query time-window $w=(e,l)$ and a correlation threshold $\theta$. The query window is defined with an end timestamp $e$ and a length $l$ that indicates a sub-sequence  of size $l$ in a time-series with a start timestamp $e-l+1$ and an end timestamp $e$. For real-time data, the end timestamp can be the last observed time, i.e. a user query on real-time data is $w=(``now", l)$, that means the network is constructed on the last $l$ observed data points. 
We consider the data points within $w$ for each time-series  $x=[\x_1, \ldots, \x_k]$. 
For example, $[\x_{k-m+1}, \ldots, \x_m]$ is the sequence we consider for $x$  on 
the query window $w=(k,m)$. 
When clear from the context, we call the sequence of a time-series $x$, for a given  query window simply query window or time-series $x$.  
The climate network  of $\mathcal{L}$ for a  given query time-window 
$w$ is a graph $\mathcal{N}=(G,V)$, 
where a node in $G$ corresponds to a location $i$ and is represented 
by time-series $x^i$. 
An edge in $V$ between nodes $i$ and $l$ indicates that the correlation between time-series $x^i$ and $x^l$ is above a user-defined threshold $\theta$.   
In this paper, we focus on the most commonly used correlation measure 
i.e. Pearson’s correlation coefficients. 
Given query windows  $x=[\x_1, \ldots, \x_m]$ and $y=[\y_1, \ldots, \y_m]$, with means $\bar{x}$ and $\bar{y}$, the Pearson's correlation of $x$ and $y$ is calculated as follows.

\begin{equation}\label{eq:corr}
{\displaystyle Corr(x, y) = {\frac {\sum _{i=1}^{m}(\x_i-{\bar{x}})(\y_i-{\bar {y}})}{{\sqrt {\sum _{i=1}^{m}(\x_i-{\bar {x}})^{2}}}{\sqrt {\sum _{i=1}^{m}(\y_i-{\bar {y}})^{2}}}}}}
\end{equation}

Constructing a network for a query window (e.g. the first six months of 2021)   requires computing the correlation for all  pairs of time-series and pruning  links using  threshold $\theta$. 
Existing techniques for fast all-pair  correlation calculation 
on large time series approximate pairwise correlation by relying on the Fourier transform~\cite{MueenNL10,ZhuS02,ColeSZ05}. 
In this paper, we present an efficient way of calculating the exact correlation of 
time-series and constructing a network for historical data and updating the network for real-time data. 
The existing work divide time-series into cooperative and uncooperative to perform correlation approximation. 
Although our core contribution is the exact calculation and the update of correlations, we also present an approximate way of calculating correlation for generic time-series (\S~\ref{sec:approx}).  
We start by giving an overview of correlation approximation.  

\subsection{Correlation Approximation Solutions} 
\label{sec:bg}

Computing the correlation of large-scale time-series is pervasively done using the notion of basic windows~\cite{ZhuS02,MueenNL10}. Time-series are processed in batches of size $B$, i.e. the stream $[{\textbf x}_1, \ldots, {\textbf x}_n]$ is equally divided into $n/B$ basic windows, 
where the $j$-th basic window contains data  $[{\textbf x}_{(j-1)*B}, \ldots, {\textbf x}_{j*B}]$. Similarly, a query window is a sequence of basic windows. A query window is assumed to be divisible by the size of basic window. 
Figure~\ref{fig:windows} presents a visualization of query and basic windows. 
The existing techniques {\em approximate}  correlation using the Discrete Fourier Transform (DFT) of basic windows.   
The DFT of a time-series $x=[\x_1, \ldots, \x_k]$ is a sequence $X=[\X_1, \ldots, \X_k]$ of complex numbers: 
\begin{equation}\label{eq:dft}
\X_f = \frac{1}{\sqrt{k}}\sum_{i=1}^{k} \x_ie^\frac{-j2\pi f i}{k}, f=1,\ldots,k
\end{equation}
where $j = \sqrt{-1}$. Computing DFT coefficients has time  complexity $O(n^2)$ in the size of basic  window. 
\draco{For normalized time series,} DFT preserves the Euclidean distance between two sequences, that is, 
\draco{$Dist(\hat{x},\hat{y}) = Dist(\hat{X},\hat{Y})$}. The approximation  techniques consider the first few DFT coefficients to capture the shape and properties of time-series. 
It has been shown that the correlation of two time-series 
can be reduced to the Euclidean distance of the the DFT coefficients of their normalized time-series~\cite{Rafiei99,ZhuS02}. 
The normalization of a basic window $x_i = [\x_1, \ldots, \x_{B}]$ is $\hat{x_i} = [\frac{\x_1-\overline{x_i}}{\sigma_i},\ldots,\frac{\x_B-\overline{x_i}}{\sigma_i}]$, where $\overline{x}_i$ and $\sigma_i$ are the mean and standard deviation of $x_i$. The correlation of two time-series can be obtained from the Euclidean distance $d(.,.)$ of their normalized series. 
\begin{equation}\label{eq:distcorr}
c_i = 1 - \frac{1}{2}d^2(\hat{x_i},\hat{y_i})
\end{equation}

For more concise notation, we denote $d_i$ to be $d(\hat{x_i},\hat{y_i})$. 
Suppose $\hat{X}_i$ and $\hat{Y}_i$ are the DFT of normalized basic windows $\hat{x}_i$ and $\hat{y}_i$, 
and $Dist_n(\hat{X}_i,\hat{Y}_i)$ is the Euclidean distance of the first $n$ DFT coefficients in $\hat{X}_i$ and $\hat{Y}_i$. Recall DFT preserves the distance between coefficients and the original time-series. Therefore, $d_i\simeq Dist_n(\hat{X}_i,\hat{Y}_i)$. The more coefficients are used (the higher $n$), the more accurate the distance and correlation becomes. 
So far, we have a way of computing the distance of basic windows. To compute the distance of query windows, $Dist_n(x,y)$, 
the existing techniques assume that the form and properties of time-series do not drastically change over a query window, i.e. basic windows have similar statistics (mean and standard deviation) to the  query window~\cite{ZhuS02,MueenNL10}. 
When the statistics do not change, $Dist_n(x,y)$ is the average of the $d_i$ on all basic windows of $x$ and $y$. 
In \S~\ref{sec:approx}, we relax this assumption and consider time series that change in form and properties over a query window, 
i.e. the statistics of basic windows are not necessarily similar to each other and the query window. 
Now, we apply Equation~\ref{eq:distcorr} on time windows, to get $Corr(x,y) \simeq 1 - \frac{1}{2}Dist_n(x,y)$. Again the higher $n$ we use, the better approximation of correlation we obtain.  

Now, we describe how $Dist_n(\hat{X},\hat{Y})$ is used to decide whether $Corr(x,y)\geq\theta$. 
Zhu and Shasha show the following relationship between correlation and the distance based on $n$ DFT coefficients of $\hat{X}$ and $\hat{Y}$. 
\begin{equation}\label{eq:corrdistthreshold}
Corr(x,y) \geq 1-\epsilon^2 \Rightarrow Dist_n(\hat{X},\hat{Y}) \leq \epsilon
\end{equation}
When using approximate techniques for climate network construction,  to get the pairs of time-series with  $Corr(x,y)\geq\theta$, 
we can compute $\epsilon = \sqrt{1-\theta}$. This allows us to prune pairs with condition $Dist_n(\hat{X},\hat{Y}) \leq \sqrt{1-\theta}$. 

Using Equation~\ref{eq:corrdistthreshold}, we get a superset of highly correlated time-series with no false negatives. As we show in Figure~\ref{fig:accuracy}, the false positives  incur  spurious edges in the network and result in an inaccurate network. These false positive edges can only be filtered at the cost of exact correlation calculation from the raw data. 
To avoid  false positives,  
\name calculates exact correlations of time-series, even faster than approximation. 

Applying Equation~\ref{eq:corrdistthreshold} requires normalizing time-series and calculating DFT coefficients. When using DFT-based approximation, the accuracy of network increases as more DFT coefficients are considered. 
Indeed, approximate techniques consider very few coefficients (two in the case of StatStream~\cite{ZhuS02} for any basic window size). 
However, when dealing with climate data sets, which are uncooperative time-series, the majority of coefficients are needed to get near accurate results (Figure~\ref{fig:accuracy}). Statstream proposes random projection for uncooperative 
time-series that similar to DFT coefficient calculation approximates correlation and has high overhead. 
To overcome modeling uncooperative time-series, Qiu et al. use Fourier transform
and neural network to embed time series into a low-dimensional Euclidean space~\cite{qiu2018learning}. 
The search is done using a nearest neighbor search  index in the embedding space. 
In this paper, we choose a different approach and compute exact correlation and network. Our solution, \name, relies on 
simple statistics of basic windows, and as we show empirically is much faster than approximation techniques. In particular, \name can update a network constructed on real-time data more efficiently than approximation techniques. This enables interactive network analysis on accurate networks for historical and real-time climate data.

\section{Network Construction} 

\begin{table}
\caption{Table of Notations}\vspace{-4mm}
\begin{tabular}{ l  l}
\hline
Symbol & Description \\
\hline
$x$ & a query window $[\x_1, \ldots, \x_m]$ of stream $\x$\\
$\x_i$ & data value at time $i$ in a query window $x$\\
$\overline{x}$ & mean of $x$\\
$x_j$ & the $j$-th basic window of $x$\\
$\overline{x_j}$ & mean of the $j$-th basic window of $x$\\
$\overline{x_{i:j}}$ & mean of basic windows $x_i, \ldots, x_j$ of $x$\\
$n_s$ & number of basic windows in a query window \\
{$ B $ }& { number of data  points in a basic window }\\
$\sigma_{xj}$ & standard deviation of the $j$-th basic window of $x$ \\
$c_j$ & correlation of $x$ and $y$ on the $j$-th basic window \\
 \hline
\end{tabular}
\label{tbl:notations}
\end{table}

Before a deep dive into exact correlation calculation, we present a high-level overview of \name's end-to-end framework. Figure~\ref{fig:arch} illustrates the components of \name for constructing and updating a climate network on historical and real-time data. 
The data storage contains a collection of frequently updated time-series accessible through locations. 
During the pre-processing, every time-series is divided into basic windows. We sketch  basic windows of time-series, in one pass, and store the collected statistics. This can be done at data ingestion time. 
At query time, the statistics of the basic windows corresponding to a given query window of all time-series are retrieved and all-pair correlation is calculated without the need to access the raw data. 
For real-time data, the system constructs the initial network and  ingests the real-time raw data in chunks of size $B$. The sketching of the newly ingested basic window is done on the fly and the correlations of time-series are updated incrementally without 
computing the correlation from scratch. 
In the following sections, we describe the details of data sketching as well as the mathematical model for the exact and approximate calculation of correlation. 

\begin{figure}[!t]
    \centering
    \includegraphics[width=\linewidth]{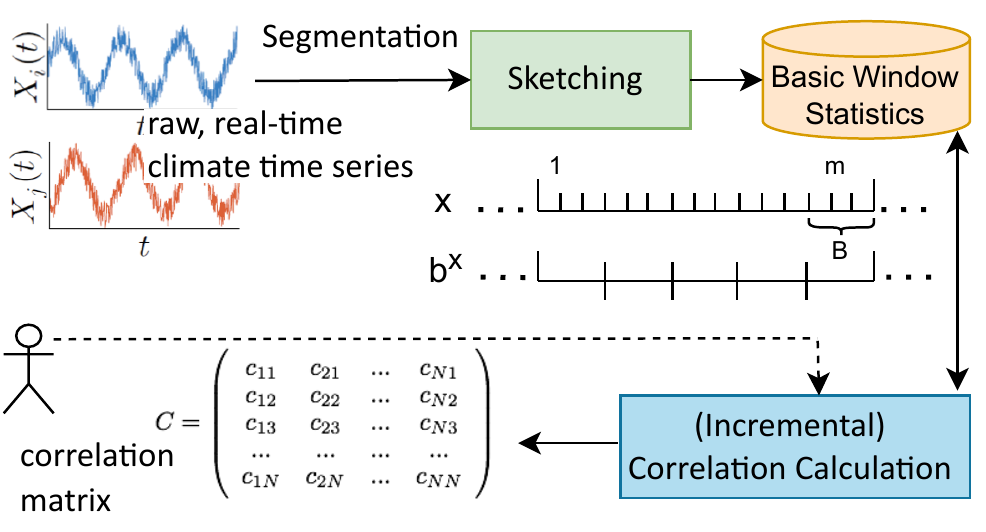}
    \vspace{-5mm}
    \caption{Architecture of \name. {update figure with parallel info.}}
    \label{fig:arch}
\end{figure}

\subsection{Exact Pairwise Correlation}
\label{sec:exact}

\subsubsection{Historical Data} 

\begin{algorithm}[t]
\caption{\preproc}\label{alg:preproc}
\begin{algorithmic}[1]
\Require {streams $\mathcal{L}=\{\x^1, \ldots, \x^n\}$; basic window size $B$}
\Ensure {statistics $S$}
\State $n_s\gets {\bf Len}(\x^1)/B$
\State $S\gets \{\}$
\For{$x, y\in\mathcal{L}$}
    \State $x\gets {\bf BasicWin}(\x,B)$; $y\gets {\bf BasicWin}(\y,B)$
    \For{$j\in [1..n_s]$}
        \State $S_{x_j}\gets {\bf Stats}(x_j)$; $S_{y_j}\gets {\bf Stats}(y_j)$
        \State $c_{j}\gets {\bf Corr}(x_j,y_j)$
        \State $\hat{x}_j\gets {\bf Normalize}(x_j, S_{x_j})$; $\hat{y}_j\gets {\bf Normalize}(y_j, S_{y_j})$
        \State $\hat{X}_j\gets {\bf DFT}(\hat{x}_j)$; $\hat{Y}_j\gets {\bf DFT}(\hat{y}_j)$
        \State $d_{j}\gets {\bf Dist_n}(\hat{X}_j,\hat{Y}_j)$
        \State {\tt // 8-10 are performed for  approximation method}
        \State $S\gets ${\bf WriteStats}($S_{x_j},S_{y_j},c_{j},d_{j}$)
    \EndFor
\EndFor
\State \textbf{return} $S$
\end{algorithmic}
\end{algorithm}
\begin{algorithm}[t]
\caption{\nethist}\label{alg:netconst}
\begin{algorithmic}[1]
\Require {streams $\mathcal{L}=\{\x^1, \ldots, \x^n\}$; statistics $S$; query $w$; basic window size $B$; threshold $\theta$}
\Ensure {graph $(G,V)$}
\State $G\gets \{1,\ldots,n\}$; $V\gets \{\}$
\State $b\gets {\bf GetBasicWins}(w)$ {\tt // basic window ids in $w$}
\For{$\x\in\mathcal{L}$ and $\y\in\mathcal{L}$}
    \State $S_{x}\gets {\bf ReadStats}(S,b,x)$; $S_{y}\gets {\bf ReadStats}(S,b,y)$ 
    \State $c\gets {\bf Corr}(S_{x},S_{y})$ // {\tt use Lemma~\ref{th:exactcorr}}
    \If{$c > \theta$}
        \State $V.{\bf Add}(x,y,c)$
    \EndIf
\EndFor
\State \textbf{return} $(G,V)$
\end{algorithmic}
\end{algorithm}

{Our solution uses the basic window model to calculate the {\em exact correlation} of times-series. Subdividing a series into basic windows allows us to process data in smaller batches. Existing works  for approximate correlation calculation assume equal length across all basic windows~\cite{ZhuS02,MueenNL10}. Assuming an equal basic window length poses a  limitation on the length of the query window, that is the length of the query window could only be an integral multiple of the length of the basic window. Formally, we have $\mid x \mid = B\cdot n_s$, where $\mid x \mid $ is the length of the basic window. \name relaxes this assumption and provides a way of calculating the exact pairwise correlation of time-series for arbitrary query window lengths. }

{
\begin{lemma}
\label{th:exactcorr}
\sloppy{Given query windows $x = [{\textbf x}_1, \ldots, {\textbf x}_m]$ and $y = [{\textbf y}_1, \ldots, {\textbf y}_m]$ and the sizes of basic windows: $\textbf{B}=[B_1,B_2,\ldots , B_m]$, 
where $B_i$ is the size of the $i$th basic window size.
The exact Pearson's correlation of $x$ and $y$ is: }
\begin{align*}
  Corr&(x,y)=\frac{\sum_{j=1}^{n_s}B_j(\sigma_{x_j} \sigma_{y_j} c_j + \delta_{x_j}\delta_{y_j})}{\sqrt{\sum_{i=1}^{n_s}B_i({\sigma_{x_i}}^{2}+{\delta_{x_i}^2})} \sqrt{\sum_{i=1}^{n_s}B_i({\sigma_{y_i}}^{2}+{\delta_{y_i}^2})}}\\
  &\delta_{x_i}=\overline{x_i}-\frac{\sum_{k=1}^{n_s}\overline{x_k}}{n_s},~~\delta_{y_i}=\overline{y_i}-\frac{\sum_{k=1}^{n_s}\overline{y_k}}{n_s}
\end{align*}
where, $\sigma_{x_i}$ ($\sigma_{y_i}$) is the standard deviation of basic window of $x_i$ ($y_i$), 
$c_{i}$ is the  correlation of  basic windows  $x_i$ and $y_i$, 
$\overline{x_i}$ ($\overline{y_i}$) is the mean of basic window $x_i$ ($y_i$). 
\label{th:exactcorrgel}
\end{lemma} }
{
\begin{proof} 
This lemma has been provided as a  possible general extension provided by Dunlap~\cite{dunlap1937combinative}, without proof. We provide a proof here. 
Let $\Omega_j$ be the size of the tail of a time-series with $B_1$ to $B_j$ basic windows of arbitrary size. 
\begin{align*}
\Omega_j = \sum_{k=1}^{j}B_k ; \Omega_0 =0
\end{align*}
\begin{align*}
  &Corr(x,y) = \frac{1}{T}\sum_{j=1}^{n_s} \sum_{i=\Omega_{j-1}+1}^{\Omega_{j}} (\frac{\x_i-\overline{x}} {\sigma_x}).(\frac{\y_i-\overline{y}}{\sigma_y})\\&
         =\frac{1}{T}\sum_{j=1}^{n_s} \sum_{i=\Omega_{j-1}+1}^{\Omega_{j}} \frac{ \sigma_{x_j} \x^{j,i} + \delta_{x_j} }{\sigma_{x} } . \frac{ \sigma_{y_j} \y^{j,i} + \delta_{y_j} }{\sigma_{y} }  \\
        \eat{ &= \frac{1}{T}\sum_{j=1}^{n_s} \sum_{i=\Omega_{j-1}+1}^{\Omega_{j}} \frac{\sigma_{x_j} \sigma_{y_j} \x^{j,i} \y^{j,i}+ \delta_{x_j} \sigma_{y_j} +\delta_{y_j} \sigma_{x_j} x^{j,i} + \delta_{x_j} \delta_{y_j}}{\sigma_{x} \sigma_{y}} \\}
         &=\frac{1}{T} \frac{\sum_{j=1}^{n_s} B_j(\sigma_{x_j} \sigma_{y_j} c_j + \delta_{x_j}\delta_{y_j}) }{ \sigma_{x} \sigma_{y} }
    \\
         &=\frac{\sum_{j=1}^{n_s}B_j(\sigma_{x_j} \sigma_{y_j} c_j + \delta_{x_j}\delta_{y_j})}{\sqrt{\sum_{i=1}^{n_s}B_i({\sigma_{x_i}}^{2}+{\delta_{x_i}^2})} \sqrt{\sum_{i=1}^{n_s}B_i({\sigma_{y_i}}^{2}+{\delta_{y_i}^2})}}
   \\
\end{align*}
where $\x^{j,i} (\y^{j,i})$ is the $\x_i(\y_i)$ normalized in the $j$-th basic window. Now, to show that 
$\sigma_{x}  =  \sqrt{\frac{1}{T}{\sum_{1}^{n_s}B_i({\sigma_{x_i}}^{2}+{\delta_{x_i}^2})}}
$, we evaluate: 
{
\begin{align*}
  & T {\sigma_{x}}^2 -{\sum_{1}^{n_s}B_i({\sigma_{x_i}}^{2}+{\delta_{x_i}^2} )} = {\sum_{1}^{n_s}({\sigma_{x}}^2-{\sigma_{x_i}}^{2}-{\delta_{x_i}^2} )}\\&
  = \sum_{k=1}^{n_s}\sum_{j=\Omega_{k-1}+1}^{\Omega_{k}}{\frac{1}{B_k} (\x_j-\overline{x})^2-\frac{1}{B_k}(\x_j-\overline{x_k})^2-\frac{1}{B_k}(\overline{x_1}-\overline{x_k})^2}\\&
  = \sum_{k=1}^{n_s} (\sum_{j=\Omega_{k-1}+1}^{\Omega_{k}} \frac{ (-2\overline{x}{\x_j} + 2{\overline{x_k}}\,{\overline{x}} ) +(2{\x_j}\overline{x_k}-2{\overline{x_k}}^2)}{B_k})\\
 &
  = \sum_{k=1}^{n_s} \frac{1}{B_k} ( (-2B_k\overline{x}{\overline{x_k}} + 2B_k{\overline{x_k}}\,{\overline{x}} ) + (2B_k{\overline{x_k}}^2-2B_k{\overline{x_k}}^2)) =0
\end{align*}}
\end{proof}
}

{
Using Lemma~\ref{th:exactcorr}, we can pre-compute and store the statistics of  basic windows of once and compute the correlation of time-series for user-given  time windows 
at query time without performing a pass over time-series. Moreover,  Lemma~\ref{th:exactcorrgel} allows us to support the arbitrary query window lengths. For instance, a  user-provided query window   $x=[x_i,\ldots,x_j]$ and $y=[y_i,\ldots,y_j]$, there exists a unique $\kappa \in \mathcal{N}$ such that $\kappa \cdot B \leq i < (\kappa +1) \cdot B $, and there exists a unique $\chi \in \mathcal{N}$ such that $\chi \cdot B < j \leq (\chi +1) \cdot B $. Let $B_1= (\kappa +1) \cdot B - i$,  $B_{n_s} = \chi \cdot B-j$, and $B_k$, for $k \in \{2,\cdots, n_s-1\}$. 
At query time, we  need to compute $\sigma_{x_1} (\sigma_{y_1}$), $\sigma_{x_{n_s}} (\sigma_{y_{n_s}}$), $\delta_{x_1} (\delta_{y_1}$) and $\delta_{x_{n_s}} (\delta_{y_{n_s}}$) from the raw data, and all the others for the $B_2, \cdots, B_{n_s-1}$ are pre-computed in the pre-processing. Note that the case of equally subdividing time-series into basic windows of size $B$ and a query window size being the integral multiple of the basic window size is a special case of Lemma~\ref{th:exactcorr}. For this special case, Algorithm~\ref{alg:preproc} shows the steps of sketching basic windows and Algorithm~\ref{alg:netconst}  describes the steps of constructing a network based on the exact correlation of time-series calculated from the pre-computed statistics of basic windows. }

\subsubsection{Real-time Data}

The correlation equation of  Lemma~\ref{th:exactcorr} can be extended to deal with real-time data. 
A user-defined query window on real-time data, $w=(``now",m)$, indicates the sequence of the $m$ most recently observed data points of time-series. That is, the size of the query window is fixed while the end timestamp is changing as new data arrive. 
In our problem setting, since the data is processed on the basis of basic windows, the algorithm waits  until all new $B$  data points arrive. 
For time-series $x=[\x_1, \ldots, \x_m]$ and $y=[\y_1, \ldots, \y_m]$, and  a query $w=(``now",m)$, we can compute correlation at time $t$, namely $Corr_t(x,y)$,  using Lemma~\ref{th:exactcorr}. 
This involves considering basic windows  $[x_1, \ldots, x_{n_s}]$ and $[y_1, \ldots, y_{n_s}]$, where $n_s=m/B$.  
At time $t+B$, the observed time-series are  $[\x_1, \ldots, \x_{m+B}]$ and $[\y_1, \ldots, \y_{m+B}]$ and the basic windows are 
$[x_1, \ldots, x_{n_s+1}]$ and $[y_1, \ldots, y_{n_s+1}]$. 
Based on query $w=(``now",m)$, we need to consider 
$[x_2, \ldots, x_{n_s+1}]$ and $[y_2, \ldots, y_{n_s+1}]$. 
According to Lemma~\ref{th:exactcorr}, 
we can recalculate the correlation at time $t+B$ from scratch. That is, 
\begin{align*}
  &Corr_{t+B}(x,y)=\frac{\sum_{i=2}^{n_s+1}(\sigma_{x_i} \sigma_{y_i} c_i + \delta_{x_i}\delta_{y_i})}{\sqrt{\sum_{i=2}^{n_s+1}({\sigma_{x_i}}^{2}+{\delta_{x_i}^2})} \sqrt{\sum_{i=2}^{n_s+1}({\sigma_{y_i}}^{2}+{\delta_{y_i}^2})}}
\end{align*}
Note that $\delta_{x_j}$'s and $\delta_{y_j}$'s have changed and needs to be recalculated, since the means of the new query windows have probably changed upon the arrival of new data. 
The following lemma allows us to compute $Corr_{t+B}(x,y)$ by only using the statistics of the first and last basic windows, without the need to calculate the statistics of  the query window.

\begin{lemma}
\label{th:realtimecorr}
\sloppy{{Given query windows $x = [\x_1, \ldots, \x_m]$ and $y = [\y_1, \ldots, \y_m]$,
basic windows  $[x_1, \ldots, x_{n_s}]$ and $[y_1, \ldots, y_{n_s}]$, and basic window size $B = [B_1, \ldots, B_{n_s}]$, where $T = \sum_{i=1}^{n_s}B_i$. 
Upon the arrival of $B_{n_s+1}$ new data points, we have $x=[\x_1, \ldots, \x_{m+B_{n_s+1}}]$ and $y = [\y_1, \ldots, \y_{m+B_{n_s+1}}]$ and basic windows 
$[x_1, \ldots, x_{n_s+1}]$ and $[y_1, \ldots, y_{n_s+1}]$. Let $T'=\sum_{i=2}^{n_s+1}B_i$. Considering a query window $w=(``now", m)$, we can incrementally compute the Pearson's correlation of $x$ and $y$ 
at time $t+B_{n_s+1}$ from their correlation at time $t$:} }
{
\begin{align*}
  Corr_{t+B_{n_s+1}}(x,y) = &\frac{1}{C~.~D}~\Big( T\sigma_x\sigma_y Corr_t(x,y) \\&+ B_{n_s+1}( \sigma_{x_{n_s+1}} \sigma_{y_{n_s+1}} c_{n_s+1} + \delta_{x_{n_s+1}} \delta_{y_{n_s+1}}) \\&- B_1(\sigma_{x_1} \sigma_{y_1} c_1+\delta_{x_1} \delta_{y_1})  {-} T' \alpha_{x} \alpha_{y}\Big)
\end{align*}}
{\[C = \sqrt{T\sigma_{x}^2+B_{n_s+1}({\sigma_{x_{n_s+1}}}^{2} + \delta_{x_{n_s+1}}^2)-B_1(\sigma_{x_1}^2+\delta_{x_1}^2)-T' \alpha_{x}^2}\]
\[D = \sqrt{T\sigma_{y}^2+B_{n_s+1}({\sigma_{y_{n_s+1}}}^{2} + \delta_{y_{n_s+1}}^2)-B_1(\sigma_{y_1}^2+\delta_{y_1}^2)-T' \alpha_{y}^2}\]
\[
  \alpha_x = \frac{B_{x_{n_s+1}}\delta_{n_s+1}-B_{1}\delta_{x_1}}{T} \quad and \quad \alpha_y = \frac{B_{n_s+1}\delta_{y_{n_s+1}}-B_{1}\delta_{y_1}}{T}\]
 \[\delta_{x_{n_s+1}} = \overline{ x_{n_s+1}}-\overline{x_{1:n_s}} \quad and \quad \delta_{y_{n_s+1}} = \overline{ y_{n_s+1}}-\overline{y_{1:n_s}}\]}
{
\noindent where, $\sigma_x$ ($\sigma_y$) is the standard deviation of query window $x$ ($y$) at time $t$, $\sigma_{x_j}$ ($\sigma_{y_j}$) is the standard deviation of the basic window of $x_j$ ($y_j$),  $c_{j}$ is the  correlation of the $j$-th basic windows of $x$ and $y$,  $\overline{x_j}$ ($\overline{y_j}$) is the mean of basic window $x_j$ ($y_j$), and $\overline{x_{i:j}}$ ($\overline{y_{i:j}}$) is the 
mean of basic windows $x_i,\ldots,x_j$ ($y_i,\ldots,y_j$). }
\label{th:realtimecorr_ge}
\end{lemma}

\sloppy{{\begin{proof} We denote $\sigma_x'(\sigma_y')$ to be the standard deviation of the new query window $w$ after the arrival of the new basic window, and $\delta_{x_i'}(\delta_{y_i'})$ to be parameters of the new query window.
We assume there is a linear transform  $Corr_{t+B_{n_s+1}}(x,y) = k Corr_{t}(x,y) + s$, where $k = \frac{T\sigma_{x}\sigma_{y}}{T'\sigma_x'\sigma_y'}$ and $s = \frac{1}{T'\sigma_x'\sigma_y'} s'$. Looking at the change of the
numerator, we have derived the following for $s'$: 
\begin{align*}
s'&=T'\sigma_x'\sigma_y' Corr_{t+B_{n_s+1}}(x,y) - T\sigma_x\sigma_y Corr_{t}(x,y) \\&=
\sum_{i=2}^{n_s+1}B_i(\sigma_{x_i} \sigma_{y_i} c_i + \delta_{x_i'}\delta_{y_i'})-\sum_{i=1}^{n_s}B_i(\sigma_{x_i} \sigma_{y_i} c_i + \delta_{x_i}\delta_{y_i})\\
&= B_{n_s+1}\sigma_{x_{n_s+1}} \sigma_{y_{n_s+1}} c_{n_s+1}- B_{1}\sigma_{x_1} \sigma_{y_1} c_{1} + \sum_{i=2}^{n_s+1} B_i\delta_{x_i'}\delta_{y_i'} \\&- \sum_{i=1}^{n_s}B_i \delta_{x_i}\delta_{y_i}
\end{align*}
Since $\delta_{x_i}'=\overline{x_i} - (\overline{x_{1:n_s}}+\alpha_x)$ and $\delta_{y_i}'=\overline{y_i} - (\overline{y_{1:n_s}}+\alpha_y)$, we have
\begin{align*}
\delta_{x_i}'\delta_{y_i}' & =(\overline{x_i} - (\overline{x_{1:n_s}}+\alpha_x))(\overline{y_i} - (\overline{y_{1:n_s}}+\alpha_y))\\
 &=(\delta_{x_i} - \alpha_x) (\delta_{y_i} - \alpha_y)
\end{align*}
Then, plugging it into the equation for $s'$, we get
\begin{align*}
s^\prime&=B_{n_s+1}(\sigma_{x_{n_s+1}} \sigma_{y_{n_s+1}} c_{n_s+1} + \delta_{x_{n_s+1}}\delta_{y_{n_s+1}}) \\&- B_1(\sigma_{x_1} \sigma_{y_1} c_1+ \delta_{x_{1}}\delta_{y_{1}}) -T'\alpha_x\alpha_y
\end{align*}
To get $s$ and $k$, we also need the incremental equation for $\sigma_x'(\sigma_y')$ if we look at the denominator. In the proof of  Lemma~\ref{th:exactcorr}, we showed that 
$\sigma_{x}= \sqrt{\frac{1}{T}{\sum_{1}^{n_s}B_i({\sigma_{x_i}}^{2}+{\delta_{x_i}^2})}}$. We have: 
\begin{align*}
T'\sigma_x'^2 &= B_2{\sigma_{x_2}}^{2}+B_2{\delta_{x_2}'^2}+\ldots+B_{n_s+1}\sigma_{_{n_s+1}}^2+B_{n_s+1}{\delta_{x_{n_s+1}}'^2}\\
            &= \sum_{j=1}^{n_s} B_j{\sigma_{x_j}}^{2} + B_{n_s+1}{\sigma_{x_{n_s+1}}}^{2} - B_1{\sigma_{x_1}}^{2} + \sum_{j=2}^{n_s +1}B_j (\delta_{x_j} - \alpha_x)^2\\
            &= \sum_{j=1}^{n_s} B_j {\sigma_{x_j}}^{2} + B_{n_s+1} {\sigma_{x_{n_s+1}}}^{2} - B_1 {\sigma_{x_1}}^{2} \\&+ \sum_{j=2}^{n_s +1} B_j ({\delta_{x_j}}^2 - 2  {\delta_{x_j}\alpha_x}+  {\alpha_x}^2)\\
            &=T\sigma_{x}^2+B_{n_s+1}({\sigma_{x_{n_s+1}}}^{2} + \delta_{x_{n_s+1}}^2)-B_1(\sigma_{x_1}^2+\delta_{x_1}^2)-T' \alpha_{x}^2 \\
          \end{align*}
By replacing $k$ and $s$ in the above  transform, the proof of the lemma becomes complete. For real-time data, the mean of the query window changes as new data arrives. For efficiency purposes, we do not want to compute that mean for the calculation of parameters $\delta$. Lemma~\ref{th:exactcorrgel} allows us to compute the standard deviation of a query window and correlation without computing the mean. 
\eat{
We provide a proof sketch of $Corr_{t+B}(x,y)$. Suppose the new correlation has the format: $Corr_{t+B}(x,y)  = \frac{n_s\sigma_x \sigma_y Corr_t(x,y) + s'}{n_s\sigma_x'\sigma_y'}$ by the equation obtaining:
\begin{align*}
& Corr_{t+B_{n_s+1}}(x,y)  = \frac{n_s\sigma_x \sigma_y Corr_t(x,y) + s'}{n_s\sigma_x'\sigma_y'} =\frac{1}{n_s\sigma_x'\sigma_y'}\Big(\\& n_s\sigma_x\sigma_y Corr_t(x,y) + \sigma_{x_{n_s+1}} \sigma_{y_{n_s+1}} c_{n_s+1}-\sigma_{x_1} \sigma_{y_1} c_1 + \delta_{x_{n_s+1}}\delta_{y_{n_s+1}} \\& - \delta_{x_{1}}\delta_{y_{1}} 
        - \sum_{j=2}^{n_s+1}(\alpha_y(\overline{x_j}-\overline{x_{1:n_s}}) + \alpha_x(\overline{y_j}-\overline{y_{1:n_s}})-  \alpha_x\alpha_y) \Big) \\&
=  \frac{1}{n_s\sigma_x'\sigma_y'}\Big( n_s\sigma_x\sigma_y Corr_t(x,y) + \sigma_{x_{n_s+1}} \sigma_{y_{n_s+1}} c_{n_s+1}\\& - \sigma_{x_1} \sigma_{y_1} c_1-\delta_{x_1} \delta_{y_1} {-} n_s \alpha_{x} \alpha_{y} {+} \delta_{x_{n_s+1}} \delta_{y_{n_s+1}}\Big)
\end{align*}}
\end{proof}
Algorithm~\ref{alg:netreal} describes the steps of constructing a network for real-time data. 
}}

\begin{algorithm}[t]
\caption{\netreal}\label{alg:netreal}
\begin{algorithmic}[1]
\Require {streams $\mathcal{L}=\{\x^1, \ldots, \x^n\}$; statistics $S$; query $w$; basic window size $B$; threshold $\theta$}
\Ensure {graph $(G,V)$}
\State $S\gets \preproc(\mathcal{L}, B)$
\State $G,V\gets \nethist(\mathcal{L}, S, w, B, \theta)$ {\tt // create initial network}
\State $b\gets[]$ {\tt // most recent basic window}
\While{}
\State $b\gets{\bf IngestData}()$
\If{{\bf Len}(b) == B}
    \State $s\gets{\bf Stats}(b)$
    \State {\bf UpdateNetwork}(G, V, s) {\tt // use Lemma 2}
    \State $b\gets[]$
\EndIf
\EndWhile
\State {\bf return} 
\end{algorithmic}
\end{algorithm}

\subsection{Approximate Pairwise Correlation} 
\label{sec:approx}

This considers non-aribtrary query window size.
So far, we presented ways of computing and updating the exact correlation of time-series. 
Now, we describe how our model can be extended to approximate the correlation of time-series over a  query window for all time-series regardless of being  cooperativeness or uncooperative. 
Equation~\ref{eq:distcorr} shows how the DFT coefficients of two time series can be reduced to the Euclidean distance of their normalized series, as described in \S~\ref{sec:bg}  
Note that, in our model, the  necessary statistics for normalization are collected during the sketch time. 

\subsubsection{Historical Data}
\label{sec:approxhist}

Recall $d_i$ is the distance of the normalized $i$-th basic windows, namely $\hat{x}_i$ and $\hat{y}_i$, $\hat{X}_i$ and $\hat{Y}_i$ are the DFT of normalized basic windows $\hat{x}_i$ and $\hat{y}_i$, 
and $Dist_n(\hat{X}_i,\hat{Y}_i)$ is the Euclidean distance of the first $n$ DFT coefficients in $\hat{X}_i$ and $\hat{Y}_i$. Since DFT preserves the distance between coefficients and the original time-series, we have $d_i\simeq Dist_n(\hat{X}_i,\hat{Y}_i)$. 
To compute the distance of query windows, $Dist_n(x,y)$, from the distances of basic windows, without any assumption about the form and properties of basic windows in a query window, we can combine the equation of Lemma~\ref{th:exactcorr} and Equation~\ref{eq:distcorr} as follows.
\begin{align*} 
  &1-\frac{1}{2}Dist_n(\hat{X},\hat{Y})^2\approx \frac{\sum_{i=1}^{n_s}(\sigma_{x_i} \sigma_{y_i} (1-\frac{d_i^2}{2}) + \delta_{x_i}\delta_{y_i})} {\sqrt{\sum_{i=1}^{n_s}({\sigma_{x_i}}^{2}+{\delta_{x_i}^2})} \sqrt{\sum_{i=1}^{n_s}({\sigma_{y_i}}^{2}+{\delta_{y_i}^2})}} 
\end{align*}
We simplify the equation and obtain an approximation of the distance of two query windows based on the distances of their basic windows. 
\begin{equation}
  Dist_n(\hat{X},\hat{Y})^2\approx2+\frac{\sum_{i=1}^{n_s}\sigma_{x_i} \sigma_{y_i} d_n(\hat{X_i},\hat{Y_i})^2 - 2\sum_{i=1}^{n_s}(\sigma_{x_i} \sigma_{y_i}+\delta_{x_i}\delta_{y_i}) } {\sqrt{\sum_{i=1}^{n_s}({\sigma_{x_i}}^{2}+{\delta_{x_i}^2})} \sqrt{\sum_{i=1}^{n_s}({\sigma_{y_i}}^{2}+{\delta_{y_i}^2})}}
\label{eq:dftcombine}
\end{equation}
When all DFT coefficients are used, i.e. $n=B$, the $\approx$ becomes $=$, turning into an exact calculation.

To perform all-pair correlation approximation in our framework, we can normalize basic windows and compute their DFT coefficients, and pairwise distances, during the sketch time (lines 8-10 of Algorithm~\ref{alg:preproc}). 
At query time, we use Equation~\ref{eq:dftcombine} to get $Dist_n(\hat{X},\hat{Y})$ and apply Equation~\ref{eq:distcorr} to obtain the correlation. 
Algorithm~\ref{alg:netapprox} describes the steps of building a network based on the approximation of correlation.   

\begin{algorithm}[t]
\caption{\netapprox}\label{alg:netapprox}
\begin{algorithmic}[1]
\Require {streams $\mathcal{L}=\{\x^1, \ldots, \x^n\}$; statistics $S$; query $w$; basic window size $B$; threshold $\theta$}
\Ensure {graph $(G,V)$}
\State $G\gets \{1,\ldots,n\}$; $V\gets \{\}$
\State $b\gets {\bf GetBasicWins}(w)$ {\tt // basic window ids in $w$}
\For{$x, y\in\mathcal{L}$}
    \State $S_{x}\gets {\bf ReadStats}(S,b,x)$; $S_{y}\gets {\bf ReadStats}(S,b,y)$ 
    \State $d_1\ldots d_{n_s}\gets{\bf ReadStats}(S,x,y)$ {\tt // distance of basic windows}
    \If{stats of basic windows $\simeq w$}
        \State $Dist\gets{\bf Average}([d_1,\ldots,d_{n_s}])$
    \Else
        \State $Dist\gets{\bf Distance}(S_{x},S_{y}, d_1..d_{n_s})$ // {\tt use Equation~\ref{eq:dftcombine}}
    \EndIf
    \If{$Dist \leq \sqrt{1-\theta}$}
        \State $V.{\bf Add}(x,y,Dist)$
    \EndIf
\EndFor

\State \textbf{return} $(G,V)$
\end{algorithmic}
\end{algorithm}

\subsubsection{Real-time Data} 
\label{sec:approxreal}

Combining Equation~\ref{eq:dftcombine} and Lemma~\ref{th:realtimecorr}, we can get the incremental update equation for approximating pairwise correlation:  

\begin{equation}\label{eq:approxreal}
\begin{split}
  & 2 - Dist_{n}^{t+B}(\hat{X},\hat{Y})
  \underset{\text{= (when n =b)}}{\approx}  \\&\frac{1}{A~.~B}~\Big( n_s\sigma_x\sigma_y Dist_{n}^{t}(\hat{X},\hat{Y}) + \sigma_{x_{n_s+1}} \sigma_{y_{n_s+1}}  (1-\frac{d_{n_s+1})^2}{2})\\& - \sigma_{x_1} \sigma_{y_1} (1-\frac{d_1^2}{2})-\delta_{x_1} \delta_{y_1} {-} n_s \alpha_{x} \alpha_{y} {+} \delta_{x_{n_s+1}} \delta_{y_{n_s+1}}\Big) 
\end{split}
\end{equation}
Here, $Dist_{n}^{t+B} $ ($Dist_{n}^{t}(\hat{X},\hat{Y})$) is the DFT Distance of the query window at time $t+B$ ($t$) using first $n$ coefficients in each basic window. 
The new distance can be obtained by calculating the pairwise distances for the last basic window $d_{n_s+1}$. 

\subsection{{Complexity Analysis}} 
\label{sec:complexity}


{In this section, we discuss the complexity analysis of 
query/sketch time and space overhead of  \name, the DFT-based algorithm, and the baseline algorithm for non-arbitrary query windows. Next, we describe the synergies of time and space with  usability.   
Suppose $N$ is the number of time-series and each time-series is in length $L$. }

{{\bf Space Complexity} The space overhead of \name  
is $\psi = \frac{L}{B} (2+\frac{N(N-1)}{2})$, where $B$ is the basic window size and $\frac{L}{B}$ is the number of basic windows since we divide a time-series evenly by default. For each basic window of a time-series, \name stores two values for the mean and the standard deviation. In addition, for aligned basic windows of all pairs of time-series, \name stores the correlation of each pair of time-series. 
As a result, the space complexity of \name is $O(\frac{LN^2}{B})$. The DFT-based approximate algorithm  stores the mean and the standard deviation for basic windows of each time-series and the distance between the first few DFT coefficients of aligned basic windows of pairs of time-series, thus, has the space complexity of $O(\frac{LN^2}{B})$. We remark that this space overhead is in addition to the storage of raw time-series for both algorithms if the raw time-series are not discarded after sketching. } 

{{\bf Time complexity} The sketch time complexity of \name is independent of query window size and is $O(L \cdot N^2)$, since \name requires calculating statistics over the aligned basic windows of all pairs of time-series. The sketch time complexity of the approximate algorithm is worse than \name and is $O(L^2 \cdot N^2)$, since the calculation of DFT coefficients for a time-series of length $L$ is $O(L^2)$ and 
coefficients are required for calculating the distance of aligned basic windows in all pairs of time-series.  
For a query window size $l^*=n_s \cdot B$, both \name and the approximate algorithm scan all basic windows, therefore, the query time complexity of \name and the approximate method are both $O(\frac{l^*}{B} \cdot N^2)$. However, the baseline algorithm scans the raw time-series and has the query time complexity of $O(l^* \cdot N^2)$. } 

{The query time complexity of real-time \name  is $O(B^* N^2)$, where $B^*$ is the size of the new coming basic window since \name needs to compute statistics for the new window. The query time complexity of the  real-time approximate algorithm is $O(B^{*2} N^2)$. The query time complexity of real-time baseline algorithm is $O(L^* \cdot N^2)$, where $L^*$ is the size of the query window size.} 

{{\bf Usability Discussion} Let $M$ be the maximum space capacity available for the storage of time-series sketches. 
Considering the above space analysis and assuming equal-size basic windows, the minimum basic window size of \name can be calculated by solving $\frac{L}{B} (2+\frac{N(N-1)}{2})\leq M$. That is, with $M$ available storage the maximum basic window size handled by \name is $\frac{L}{M} (2+\frac{N(N-1)}{2})$. 
Note that both time and space complexity reduce as $B$ increases. Moreover, choosing a large $B$ means less space capacity requirement.   Therefore, should we just choose an extremely large $B$? The answer is no. When an  arbitrary query window  is not supported, a large $B$ will reduce the flexibility of query windows, thus, usability. For the case of the query window size being the integral multiple of the basic window size, the chosen query window size by users becomes extremely limited. 
If we consider the generic case of  Lemma~\ref{th:exactcorr}, we will observe a significant rise in query time, since the start/end of a query window can fall anywhere in a basic window, thus, when basic windows are large, the first and last basic windows can be potentially large. Suppose the query window is in length $l^*$, where $\exists n_s \in R$, such that $n_s \cdot B \leq l^* < (n_s+1) \cdot B $. The time complexity is $O((\frac{l^*}{B}+B) \cdot N^2)$. When $B>\sqrt{l}$, $\frac{l^*}{B}+B$ is monotonically increasing. Since $B>\sqrt{l}$ at the most meaningful queries, the query time increases when the $B$ increases for the generic method.}

\subsection{{Parallel and Disk-based \name}}
\label{sec:parallel}

{The disk-based \name stores sketches on the disk to be retrieved at query time for correlation calculation. 
Moreover, despite the quadratic complexity of the sketch time and query time, \name is embarrassingly parallelizable. The set of all pairs of time-series can be partitioned into groups that are processed in parallel. 
During sketching, workers are divided into a database worker, that writes statistics to the database, and computation workers, that perform sketch computation. 
Each worker sketches time-series pairs of a partition and sends the sketches in batches to the database worker to  write to a disk-based database. During the query time, each worker is assigned a partition, reads the sketches of time-series in batches directly from the  database and computes the pairwise correlations, and outputs a sub-matrix of the correlation matrix.}

{
To leverage data locality and minimize the number I/Os, for partitioning time-series pairs, \name adopts an approach similar to the parallel block nested loop join.  Each partition contains a subset of time-series paired with all time-series. i.e. each partition is a group of rows in a correlation matrix and the processing is done row by row in batches.  
Batches of pairs are assigned to a worker and once a worker is finished, it reads the statistics of the next batch of pairs from the database. 
Since Pearson's correlation is a symmetric measure, \name needs to process 
$n(n-1)/2$ pairs to construct the correlation matrix. 
For load balancing, \name assigns the same number of pairs to each worker. Note that the same architecture can be used to make the machinery described, in \S~\ref{sec:approx}, for correlation approximation. }

\subsection{Solutions to threshold correlation matrix}
Suppose that we know $c_{xz}$ and $c_{yz}$, and we want to infer the range of $c_{xz}$. Referring to ~\cite{glass1970geometric}, we have the upper bound and the lower bound of $c_{xz}$:
\begin{equation}\label{eq:bound}
\begin{split}
  c_{xz}c_{yz}-\sqrt{(1-c_{xz}^2)(1-c_{yz}^2)} \leq c_{xy} \leq  c_{xz}c_{yz}+\sqrt{(1-c_{xz}^2)(1-c_{yz}^2)}
\end{split}
\end{equation}
The equation above provides us a way in predicting the whole correlation matrix based on a small amount of correlations. For example, we are given $N$ time series and $c_{12}, \cdots ,c_{1N}$. We could make prediction on $c_{ij}$ for any $1 \leq i,j \leq N$ from $c_{1i}$ and $c_{1j}$ based on Equation~\ref{eq:bound}. For example, given a positive threshold $\theta$, let $U_{xy:z}$ be the upper bound of $c_{xy}$ estimated from $c_{xz}$ and $c_{yz}$, and $L_{xy:z}$ be the lower bound of $c_{xy}$ estimated from $c_{xz}$ and $c_{yz}$. If $L_{xy:z} \geq \theta$ or $U_{xy:z} \leq -\theta$, then we know that $m_{xy}=1$, where $m_{ij}$ is the $i$ th row and $j$ th column of the output matrix. If $L_{xy:z} \geq -\theta$ and $U_{xy:z} \leq \theta$, then we know that $m_{xy}=0$. 

\begin{algorithm}[t]
\caption{\pruning}\label{alg:pruning}
\begin{algorithmic}[1]
\Require {streams $\mathcal{L}=\{\x^1, \ldots, \x^N\}$;threshold $\theta$}
\Ensure {Matrix $M(m_{ij})$}
\State $m_{ij}=-\infty $
\For{\texttt{$i = 1,2, \cdots, N$}}
\If{$\exists p,q \textbf{s.t.} {m_{pq}} <0 $}
   \State $c_{i1}, \cdots ,c_{iN} \gets \textbf{Computecorr}(\mathcal{L},i)$
   \For{\texttt{$j = 1,2, \cdots, N$}}
        \For{\texttt{$k = 1,2, \cdots, N$}}
        \State {$L_{jk:i},U_{jk:i} \gets \textbf{Correct-Inference}(c_{ij},c_{ik})$}
        \If{$L_{jk:i} \geq \theta$ \texttt{or} $U_{jk:i} \leq -\theta$}
        \State $m_{jk} \gets 1$
        \EndIf
        \If {$L_{jk:i} \geq -\theta$ \texttt{and} $U_{jk:i} \leq \theta$}
        \State $m_{jk} \gets 0$
        \EndIf
        \EndFor
   \EndFor
\EndIf
\EndFor
\State $M \gets\textbf{Compute-Rest}$($M$)
\State \textbf{return} $M$
\end{algorithmic}
\end{algorithm}

\begin{figure}[!ht]
    \centering
    \includegraphics[width=\linewidth]{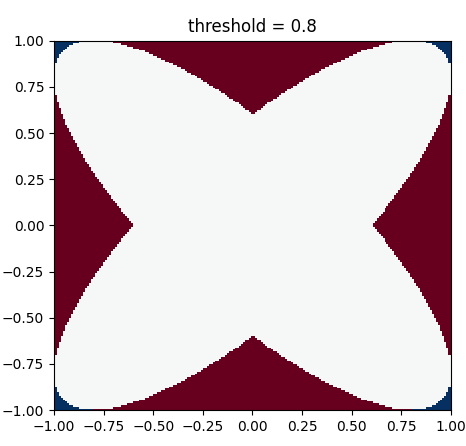}
    \vspace{-5mm}
    \caption{Correlation Inference}
    \label{fig:bound}
\end{figure}

In the Figure ~\ref{fig:bound}, let the horizontal axis be the $c_{xz}$, and the 
vertical axis be the $c_{yz}$. If $\theta$ is given to be $0.8$, then the colored regions are the cases that we could know the $m_{xy}$ without computing $c_{xy}$. In the blue regions, $m_{xy}=1$, and $m_{xy}=0$ in the red regions. The white region is the uncertainty part, we would need to change a $z$ (anchor) to see whether we could infer/compute further. The algorithm ~\ref{alg:pruning} presents how we can use it to prune. We select the anchor randomly (or let every time-series to be the anchor), then if we use the selected anchor to scan the $M$ and see which $m_{jk}$ we could know by the bounds instead of computing $c_{jk}$.

\section{Experiments}
\label{sec:experiments}

We have developed, in this paper, mathematical models and algorithms: \nethist and \netreal, for constructing and updating correlation matrices to build exact networks on historical and real-time data. {Our empirical evaluation has two parts.  First, we study these algorithms and compare their query time and sketch time against a baseline, on historical and real-time version of a climate data set.  
For these experiments, we use the in-memory version of the algorithms, i.e. in-memory data structures are used for storing raw data and sketches. 
Second, we evaluate the scalability and efficiency of the disk-based and parallel \name and the approximate algorithm as described in \S~\ref{sec:parallel}. 
For all experiments, we assume equal basic window sizes in time-series. All algorithms are implemented using Go language. We use PostgreSQL for storing data sketches.} 
{All experiments are conducted on a machine
with 2 Intel\textsuperscript{\textregistered} Xeon Gold 5218 @ 2.30GHz (64 cores), 512 GB DDR4 memory, a Samsung\textsuperscript{\textregistered} SSD 983 DCT M.2 (2 TB). }

{{\bf {\em NCEA Data Set}}\footnote{https://www.ncei.noaa.gov/pub/data/uscrn/products/hourly02/2020/} is a public data from the National Oceanic and Atmospheric Administration (NOAA). The data is collected every hour, and uploaded publicly in 24-hour increments. NOAA utilizes radiometric satellite collection, buoys, weather stations, citizen scientists, and other methods for perpetual data gathering. The data is collected from 157 nodes (time-series) across the US. Each node produces approximately 8,760 points of data in a year. 
This data set is used for in-memory experiments.}

{{\bf {\em Berkeley Earth Data Set}}\footnote{http://berkeleyearth.org/data/} is a collection of open-source data sets provided by an independent U.S. non-profit organization (Berkeley Earth). We use NetCDF-format gridded data from this data set. The climate data includes average temperature data on both lands and oceans. It divides the earth by 1$^{\circ}$ $\times$ 1$^{\circ}$ latitude-longitude grid. 
We consider the land time-series in this data set. 
The data set includes 18,638 nodes and each nodes contains 3,652 data points. The time resolution is 24 hours. This data set is used for scalability experiments.}

\subsection{Accuracy}
\label{sec:accuracy}
{We compared the accuracy of the climate network of NCEA data set, constructed based on the correlation matrix computed by the DFT-based techniques~\cite{ZhuS02,MueenNL10} (as described in \S~\ref{sec:bgapprox}) and exact calculation,  followed by the application of a threshold.}  
The approximate technique~\cite{ZhuS02,ColeSZ05} uses the first few DFT coefficients for estimating the distance of aligned basic windows, then, basic window distances are aggregated to obtain an approximation of the distance and correlation of time-series on a query window~\cite{ZhuS02}. 
In our experiments, we use the way, we believe, 
StatStream~\cite{ZhuS02} computes  
the distance (correlation) of query windows 
i.e. by averaging the distance (correlation) of DFT coefficients over all basic windows. 

We evaluate the impact of approximation on the  accuracy of constructed networks, using two measures:  
number of edges and the correlation similarity ratio, inspired by~\cite{napoletani2008reconstructing}. A correlation matrix is an $n \times n$ matrix, where $n$ is the number of time-series and a cell $c_{ij}$ is a binary value that indicates the correlation score of time-series $x^i$ and $x^j$ is higher than threshold $\theta$. The correlation similarity ratio evaluates  the percentage of identical edges in two networks. Formally, 
given two complex networks represented by
 adjacency matrices $A:\{a_{ij} \mid  0 \leqslant i , j  \leqslant n \}$ and $B :\{b_{ij} \mid  0 \leqslant i , j  \leqslant n\}$, the similarity ratio is defined as follows. 

\[D_p(A,B) =  \frac{2\sum_{i=1}^{n-1} \sum_{j=i+1}^{n} 1-\vert a_{ij}-b_{ij}\vert}{n(n-1)}\]

For instance, the correlation similarity ratio of networks with the adjacency matrices $A$ and $B$ is  $2/3$.
\[
A = 
\begin{pmatrix}
  1 & 1 & 0\\
  1 & 1 & 1\\
  0 & 1 & 1 \\
\end{pmatrix} 
\quad
B = 
\begin{pmatrix}
  1 & 0 & 0\\
  0 & 1 & 1\\
  0 & 1 & 1 \\
\end{pmatrix} 
\quad D_p(A,B) =  \frac{2}{3} \] 

\eat{
\begin{figure}[!ht]
    \centering
    \includegraphics[width=\linewidth]{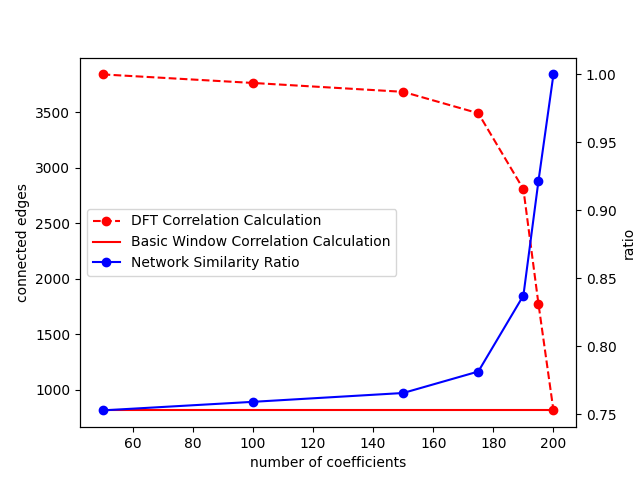}
    \caption{Accuracy Comparison.}
    \label{fig:accuracy}
\end{figure}
}

\begin{figure*}[!ht] 
    \begin{subfigure}[t]{0.245\linewidth}
        	\centering
        	\includegraphics[width =\textwidth]{figs/accuracy.png} 
        	\vspace{-5mm}\caption{}
            \label{fig:accuracy}
    \end{subfigure}
    \hfill
    \begin{subfigure}[t]{0.245\linewidth}
        	\centering
        	\includegraphics[width =\textwidth]{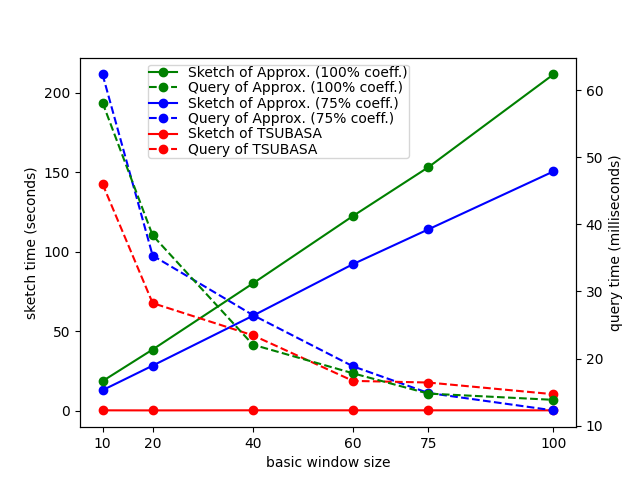} 
        	\vspace{-5mm}\caption{}
            \label{fig:constbasic}
    \end{subfigure}
    \hfill
    \begin{subfigure}[t]{0.245\linewidth}
        \centering
        	\includegraphics[width =\textwidth]{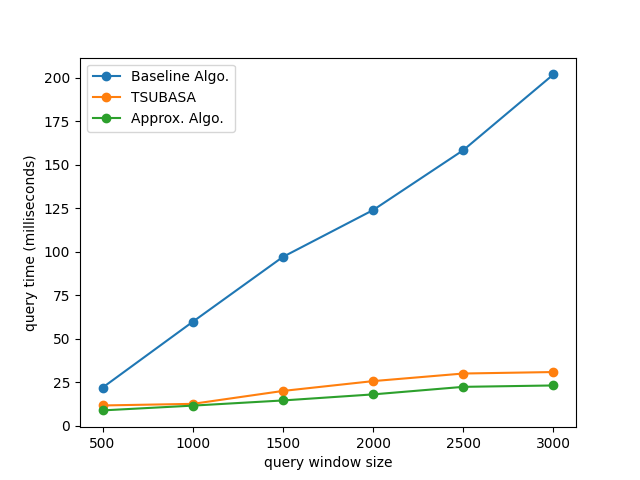} 
        	\vspace{-5mm}\caption{}
            \label{fig:constquery}
    \end{subfigure}
    \hfill
    \begin{subfigure}[t]{0.245\linewidth}
        \centering
        	\includegraphics[width =\textwidth]{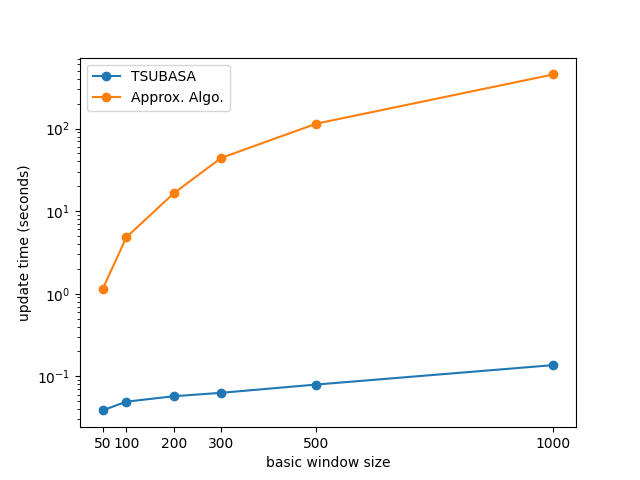} 
        	\vspace{-5mm}\caption{}
            \label{fig:update}
    \end{subfigure}
    \vspace{-3mm}
    \caption{In-memory: (a) Network Accuracy Comparison (b) Basic Window Size Analysis (c) Query Window Size Analysis (d) Network Update Time. }
    \label{fig:inmemory}
\end{figure*}

\eat{
\begin{figure*}[!ht] 
    \begin{minipage}[t]{0.33\linewidth}
        	\centering
        	\includegraphics[width =\textwidth]{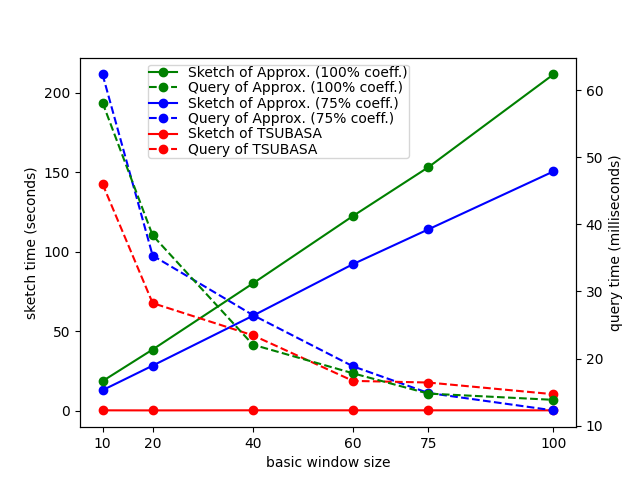} 
        	\vspace{-5mm}\caption{Sketch and Query Time vs.\\Basic Window Size.}
            \label{fig:constbasic}
    \end{minipage}
    \hfill
    \begin{minipage}[t]{0.33\linewidth}
        \centering
        	\includegraphics[width =\textwidth]{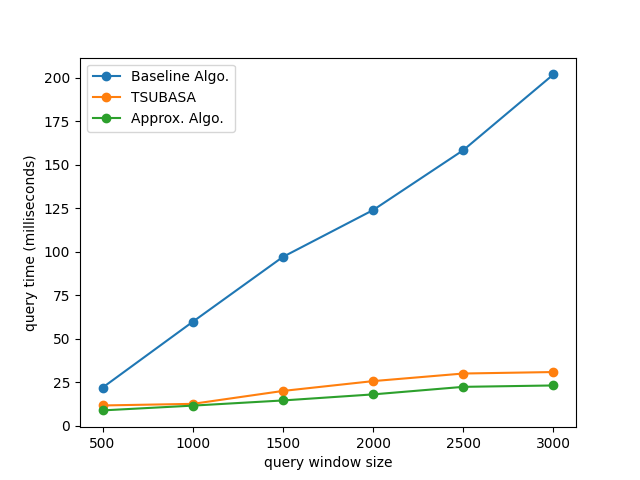} 
        	\vspace{-5mm}\caption{Query Time vs.\\Query Window Size.}
            \label{fig:constquery}
    \end{minipage}
    \hfill
    \begin{minipage}[t]{0.33\linewidth}
        \centering
        	\includegraphics[width =\textwidth]{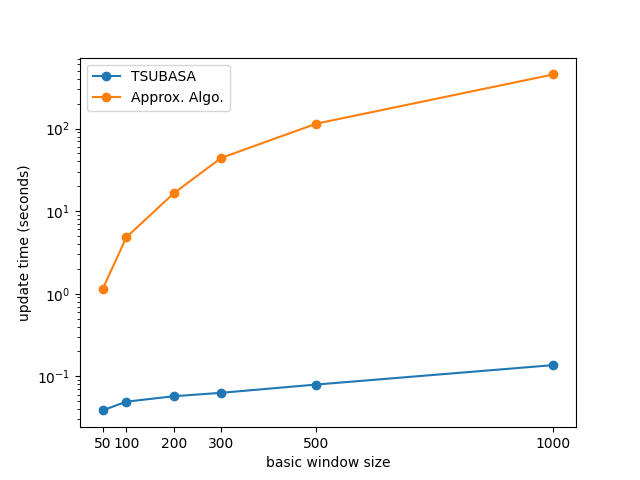} 
        	\vspace{-5mm}\caption{Network Update Time\\for Real-time Data.}
            \label{fig:update}
    \end{minipage}
\end{figure*}
}

For both techniques we consider the basic window size 200 and threshold 0.75, while varying the number of DFT coefficients from 50 to 200 for the approximate technique. 
Note that in the exact technique (basic window correlation), the correlation of time-series are computed by aggregating the correlation of basic windows as suggested by~\cite{ZhuS02}.  Therefore, the structure of this network (the solid red plot) is independent of the number of DFT coefficients. 

 As shown in Figure~\ref{fig:accuracy}, the number of edges in the network constructed by DFT correlation calculation  
becomes equal to the number of edges in the network constructed by exact calculation,   
only when all 200 coefficients are used. 
{This matches the theory, i.e. 
the approximation becomes identical to the exact calculation, when all DFT coefficients are used.}
Note that the approximate technique uses Equation~\ref{eq:corrdistthreshold} to find correlated time-series based on their   DFT-based distance. {Following this rule, the DFT correlation calculation never yields false negatives, but, creates false positive edges.} This explains why the number of edges in networks decrease as more coefficients are used.  
Moreover, the  similarity ratio of correlation matrices increases as the number of considered coefficients increases and is at its 
highest value when all coefficients are used to represent a basic window. 

The main take-away is that constructing a network based on the  approximation of DFT-based distance can lead to 
inaccurate networks. For climate data sets, near exact result is obtained only when a very large number of coefficients are used for approximation. This means smaller basic windows are preferred for approximation purposes which leads to higher number of basic windows, therefore, higher correlation calculation time in addition to the high DFT coefficient calculation time. 
These results highlight the necessity of efficient algorithms for constructing and updating exact networks on large collections of time-series. 

\subsection{{Efficiency}}
\label{sec:efficiency}

{We evaluate the efficiency of the in-memory version of correlation matrix calculation algorithms with respect to query window size and basic window size parameters. }

\textbf{Network Construction} We compare the sketch time plus query time when using the DFT-based approximation of StatStream with \name's exact correlations. 
For the approximation technique, we report on two scenarios: using all DFT coefficients and using 75\% of coefficients of a basic window. As shown in \S~\ref{sec:accuracy}, 
the former empirically yields a network similar to the network of exact correlation calculation.  
During the sketch time, \name calculates the statistics from Lemma~\ref{th:exactcorr}  and the  approximation algorithm calculates the statistics necessary for Equation~\ref{eq:dftcombine} for all basic windows of all  time-series. 
At query time, Lemma~\ref{th:exactcorr} and Equation~\ref{eq:dftcombine} are used to combine sketched statistics to get approximate and exact networks, respectively. 

Figure~\ref{fig:constbasic} reports the run time when varying the size of basic window for a query window of size {3,000}. 
{The sketch time of \name grows very gradually with the basic window size, while the sketch time of the approximate algorithm increases with 
the size of basic window.  This is because of the $O(n^2)$ complexity of DFT calculation.  }
Our results show that \name outperforms the approximation technique at sketch time and its query time is on par with the approximate network construction technique. 

Figure~\ref{fig:constquery} shows the query time of \name, approximate calculation, and a baseline, when varying the query window size, considering the constant basic window size of \fn{50}. 
The baseline algorithm computes the Pearson's correlation of Equation~\ref{eq:corr} for all pairs of time-series directly from raw data at query time  without any sketching. 
In this experiment the approximate algorithm uses 75\% of DFT coefficients of a basic window.    
Note that the distances of basic windows ($d_j$'s in Equation~\ref{eq:dftcombine}) are calculated during the sketch time, therefore,  the query time of the approximate algorithm does not depend on the number of considered DFT coefficients. 
\name  is almost as fast as the approximate algorithm
 for all query window sizes and outperforms  the baseline by two orders of magnitude. 
We remark that all algorithms have quadratic complexity in the number of time-series. 
However, the exact and DFT-based approximation are extremely efficient at computing  correlation of each pair at query time due to relying on  statistics that are pre-calculated during the sketch time.

\textbf{Network Update} We compare the network update time of \name with the  DFT-based approximation of \S~\ref{sec:approxreal} for real-time NCEA data set. The initial networks are constructed on the data set for a given query window. {Then, after the arrival of $B$ data points,  
both algorithms update the correlation and network using the special case of Lemma~\ref{th:realtimecorr} and Equation~\ref{eq:approxreal}, respectively.  Figure~\ref{fig:update} shows the  time taken by \name and approximate algorithm upon the arrival of $B$ new data points for various basic window sizes for a query window size of 3,000. }  
The approximate algorithm uses 75\% of DFT coefficients in a basic window.  
For both algorithms, updates to the network depend on the statistics of the first basic window of the current query window, which is already calculated, and the most recently observed basic window, which needs to be calculated on the fly.   
Indeed, the efficiency of update only depends on the processing of the most recently observed basic window. 
Since the approximation algorithm needs to calculate DFT coefficients we observe that it is slower than \name at least one order of magnitude. The gap between the two algorithms becomes more obvious for larger basic window sizes because of the $O(n^2)$ complexity of DFT calculation. In conclusion, \name can compute exact correlation and networks for real-time much faster than the approximation  competitor.   

\begin{figure*}[!ht] 
    \begin{subfigure}[t]{0.245\linewidth}
        	\centering
        	\includegraphics[width =\textwidth]{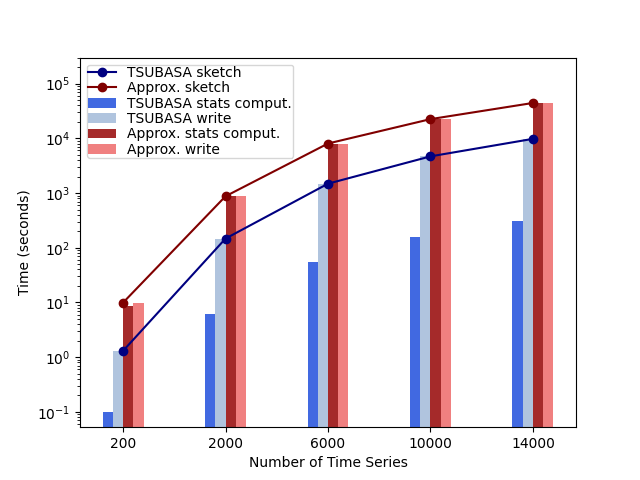} 
        	\vspace{-5mm}\caption{}
            \label{fig:parallelsketch}
    \end{subfigure}
    \hfill
    \begin{subfigure}[t]{0.245\linewidth}
        	\centering
        	\includegraphics[width =\textwidth]{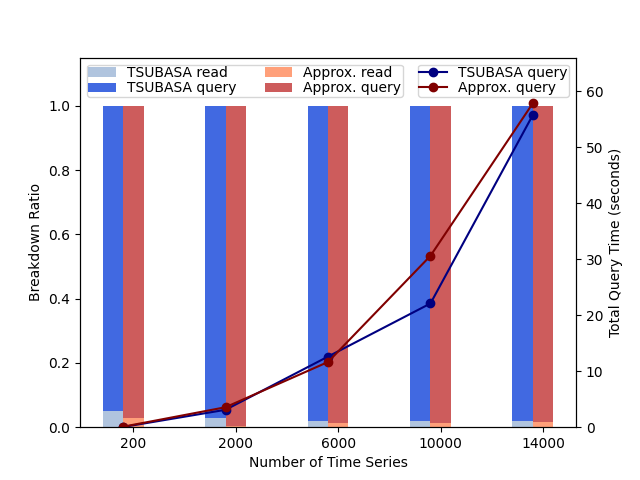} 
        	\vspace{-5mm}\caption{}
            \label{fig:parallelquery}
    \end{subfigure}
    \hfill
    \begin{subfigure}[t]{0.245\linewidth}
        \centering
        	\includegraphics[width =\textwidth]{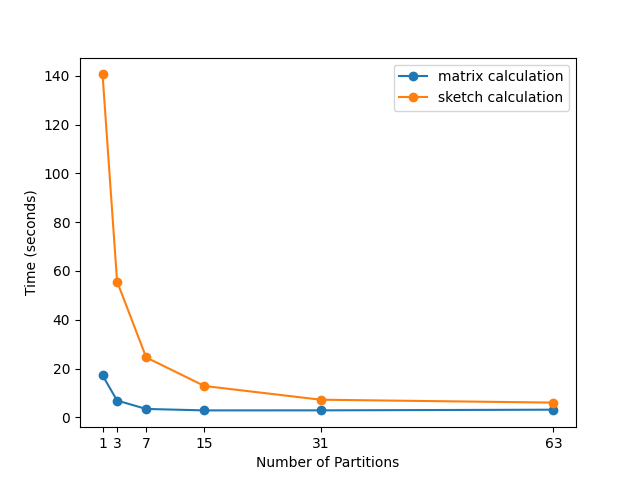} 
        	\vspace{-5mm}\caption{}
            \label{fig:scalability}
    \end{subfigure}
    \hfill
    \begin{subfigure}[t]{0.245\linewidth}
        \centering
        	\includegraphics[width =\textwidth]{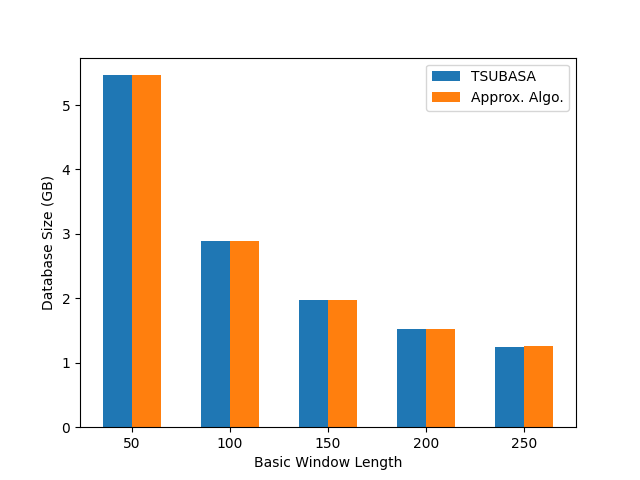} 
        	\vspace{-5mm}\caption{}
            \label{fig:space}
    \end{subfigure}
    \vspace{-3mm}
    \caption{Parallel and Disk-based: (a) Sketch Time Breakdown (b) Query Time Breakdown (c) Impact of Parallelization (d) Space Overhead. }
    \label{fig:paralleldisk}
\end{figure*}

\subsection{{Scalability}}
\label{sec:scalability}

{We compare \name and the approximation algorithm in similar parallel and disk-based configurations.   To separate the impact of fine-tuning the database on  performance, in all experiments, we choose to use one database worker and allocate the rest of workers for sketching and querying. For the scalability experiments, we use subsets of time-series from the Berkeley Earth data set. All experiments consider a basic window length of 120,  a query window length of 960, and 75\% of DFT coefficients for correlation approximation.}

{\textbf{Sketch Time} Figure~\ref{fig:parallelsketch} shows  the sketch time of \name and the approximate algorithm for correlation matrix calculation for various number of time-series on 63 partitions and 64 cores.  
The plot separates the write time from the sketch calculation time. We observe that \name outperforms the approximate algorithm in sketch time. This is due to quadratic complexity of DFT calculation as opposed the linear complexity of computing \name sketches. 
We observe that the majority of work by \name during sketching is spent on writing sketches to a database, unlike matrix approximation which is on par with the write time. Note that in this configuration the total sketch time of \name and the approximate algorithm is bounded by database write time.    
The total sketch time, sketch calculation, and write time of \name and the approximate algorithm increase  quadratically with the number of time-series.  However, due to parallelization, the growth is slower than what is expected for a single-core configuration. }

{\textbf{Query Time} Figure~\ref{fig:parallelquery}  shows the query time of \name and the approximate algorithm for correlation matrix calculation for various number of time-series on 63 partitions and 64 cores.  
The plot separates the database read time from the correlation matrix calculation time. Both \name and the approximate algorithm have on par query time and take less than a minute for computing the correlation matrix even for the largest number of time-series. We observe that the read time during querying is negligible compared to matrix calculation. The read time percentage is slightly higher for smaller networks due to the database overhead compared to matrix calculation cost on small number of time-series. 
The total query time, matrix calculation, and read time of \name and the approximate algorithm increase  quadratically with the number of time-series.  However, due to parallelization, the growth is slower than what is expected for a single-core configuration. }

{\textbf{Impact of Number of Partitions}  Figure~\ref{fig:scalability} shows how \name scales with the number of partitions. For these experiments, we use 2,000 time-series. Note than we have 63 partitions/cores for sketch/matrix computation and we reserve one core for database writes. Both sketch and matrix calculation times decrease with the increase in the number of cores. We expect further optimization of query and sketch time can be done by fine-tuning the database and allocating further resources. }

{\textbf{Space Overhead} Figure~\ref{fig:space} shows the size of databases used for storing sketches of 2,000 time-series by \name and the approximate algorithm with respect to basic window size. Both algorithms store sketches of the the same size for each basic window and have the same space overhead. As the size of basic window increase the number of basic windows decreases and total size of sketches stored by both algorithms decrease. }

\eat{
\begin{figure*}[!ht] 
    \begin{minipage}[t]{0.33\linewidth}
        	\centering
        	\includegraphics[width =\textwidth]{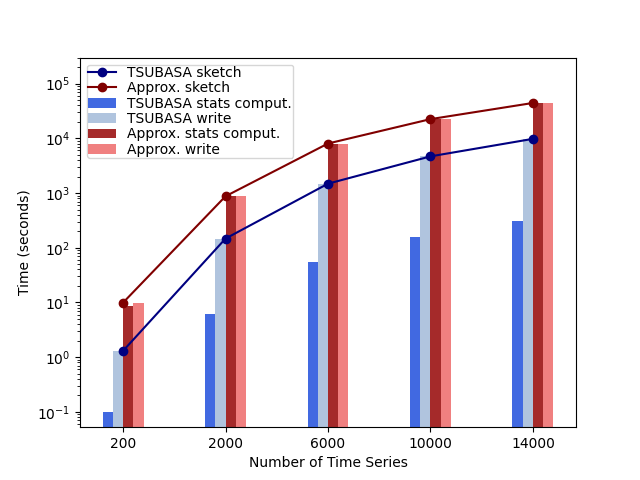} 
        	\vspace{-5mm}\caption{\fn{Sketch Time Breakdown.}}
            \label{fig:parallelsketch}
    \end{minipage}
    \hfill
    \begin{minipage}[t]{0.33\linewidth}
        \centering
        	\includegraphics[width =\textwidth]{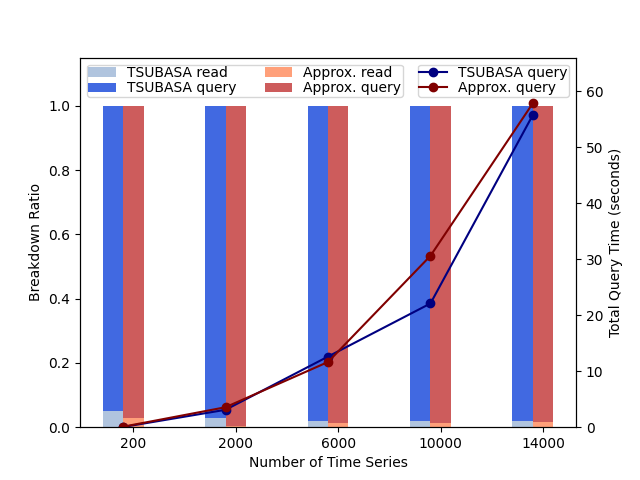} 
        	\vspace{-5mm}\caption{\fn{Query Time Breakdown.}}
            \label{fig:parallelquery}
    \end{minipage}
    \hfill
    \begin{minipage}[t]{0.33\linewidth}
        \centering
        	\includegraphics[width =\textwidth]{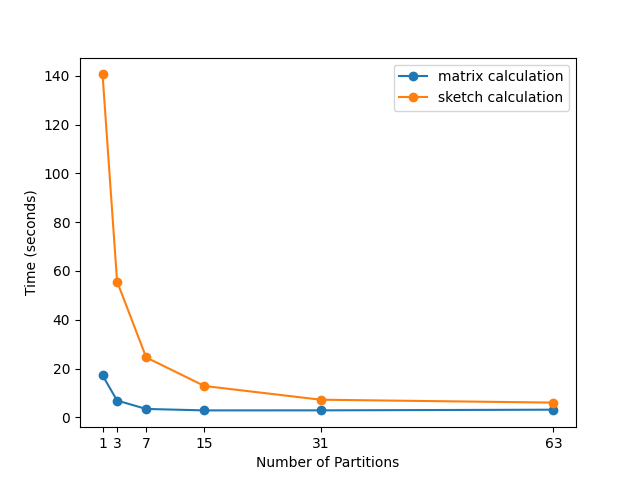} 
        	\vspace{-5mm}\caption{\fn{Impact of Number of Cores.}}
            \label{fig:scalability}
    \end{minipage}
\end{figure*}

\begin{figure}[!ht]
    \centering
    \includegraphics[width=\linewidth]{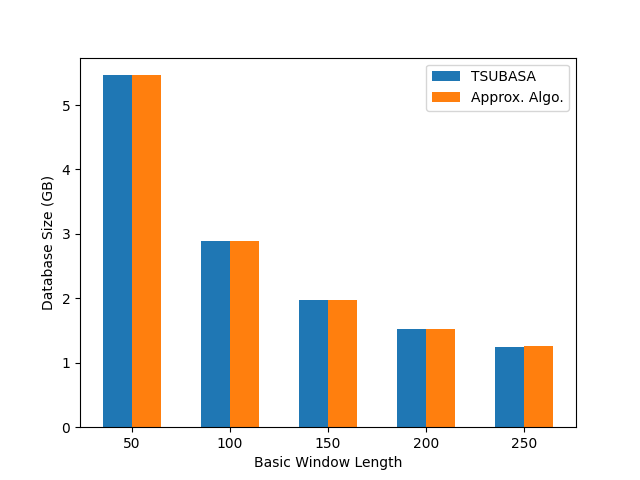}
    \caption{Space Overhead.}
    \label{fig:space}
\end{figure}
}

\section{Related Work}
\label{sec:related}

{\bf Spatio-Temporal Databases}  represent the value of a climate variable with three dimensions of geometry and time (i.e. latitude, longitude, and timestamp). 
Systems such as Microsoft StreamInsight~\cite{AliCRK10}, GeoMesa~\cite{geomesa}, and IBM PAIRS Geoscope~\cite{klein2015pairs} are designed for processing streams of geospatial data,  from sources such as satellites and  IoT sensors, 
benefiting from relational DBMSs, distributed column-oriented databases, and scalable key-value data stores. 
The algorithms and mathematical models we designed in this paper can be incorporated into geo-spatial systems. 
We use the progressive and declarative processing of  Trill~\cite{ChandramouliGBT15} for storage and analysis. \name can be a stand-alone system with other network analysis extensions for clustering and community detection.  

\eat{Database community has done extensive research in time-series similarity search. 
In particular, threshold queries for similarity search in time-series databases consider calculating a function over multiple  streams and alerting when the threshold is crossed~\cite{CaiB0C21,GogolouTEBP20,MirylenkaDP17}. For example, Cai et al. show that for quasiconvex functions, it is sufficient to retain  very few tuples per-time-series for any given query  window and never miss the event when functions does not satisfy threshold requirements~\cite{CaiB0C21}. 
For the nearest neighbor search problem, Gogolou et al.  propose a probabilistic learning-based method with quality guarantees for progressive query answering~\cite{GogolouTEBP20}.  
The existing time-series similarity search problems in the community have mostly focused on identifying similar time-series to a query time-series. 
Whereas, climate network construction requires all-pair time-series similarity calculation, independent of correlation threshold. 
In fact, constructing a complete  correlation matrix with exact correlations is a strong tool for climate network science. }

\noindent{\bf Similarity Search on Time-Series} To compute the similarity between time-series, 
several measures have been proposed~\cite{MirylenkaDP17,PaparrizosLEF20}. In \name, we consider the Pearson's correlation coefficient as it is the most commonly used measure for building climate networks~\cite{Faghmous:2014,Tsonis2004}. 
\eat{Threshold queries for similarity in time-series databases often involve calculating a function over two or more streams and reporting when the threshold is crossed. There has been extensive work from the database  community on this topic~\cite{CaiB0C21,GogolouTEBP20,MirylenkaDP17}. }
There has been extensive work from database community on similarity search of time-series~\cite{CaiB0C21,GogolouTEBP20,MirylenkaDP17,ZhuS02,qiu2018learning}.  
This line of work considers threshold queries for similarity search in time-series databases and often involve calculating a function over two or more streams and reporting when the threshold is crossed. 
\eat{For example, Cai et al. show that for quasiconvex functions, it is sufficient to retain  very few tuples per-time-series for any given query  window and never miss the event when functions does not satisfy threshold requirements~\cite{CaiB0C21}. 
For the nearest neighbor search problem, Gogolou et al.  propose a probabilistic learning-based method with quality guarantees for progressive query answering~\cite{GogolouTEBP20}. } Time-series similarity search problems in the community have mostly focused on identifying similar time-series to a query time-series. 
\name, however, focuses on the construction of the complete and exact  correlation matrix, a task that requires all-pair correlation calculation.

{\bf Sketching and Data Reduction}  Alternative techniques to DFT for time-series similarity approximation are Discrete Wavelet Transform (DWT)~\cite{ChaovalitGKC11}, Singular Value Decomposition~\cite{JangCJK18}, and Piecewise Constant Approximation~\cite{LuoYCLFHM15}. 
\draco{Data reduction is based on the idea of summarizing the population or sample data through smaller-sized matrices or simple numbers~\cite{khattree2000multivariate}. As for time-series, it has been a topic of interest that reduces data into low-dimensional data while preserving its characteristics to a large extent~\cite{zhao2006high}, which has  IoT  applications~\cite{papageorgiou2015real}. }
\name sketches data into statistics that are required for the efficient and exact calculation of correlation scores and networks. 

{\bf Data Streaming Systems} 
\jinshu{Most data streaming systems, including S4~\cite{NeumeyerRNK10}, Muppet~\cite{LamLPRVD12}, Spark~\cite{ZahariaXWDADMRV16}, and Flink~\cite{CarboneKEMHT15} are large-scale data processing systems driven by th Reduce programming model. They can work on different kinds of data coming from real-world sensors or IoT devices. The system we used in the experiments, Trill, also has a distributed version called  Quill~\cite{ChandramouliFGE16}. The distributed design will definitely reduce the latency of data processing, and also increase the throughput in a given time.} 

\section{Discussion and Conclusion}
\label{sec:conclusion}

We presented \name, an efficient and exact correlation matrix calculation  algorithm for climate network construction on historical and real-time data.  
\name uses the basic-window model to subdivide a query time-window into smaller windows. \name computes a cheap and simple sketch of basic windows and reuses them at query time for building networks on arbitrary query windows.  
We describe a way of approximating  time-series correlation and compared it with \name. Experiments show that \name can compute exact correlation and network  faster than DFT-based approximation techniques. {The techniques proposed in TSUBASA can be potentially applied to analyzing stock market data~\cite{memon2019structural} as well as biological data~\cite{batushansky2016correlation}. However, our preliminary investigation shows that approximate and partial calculation of pairwise correlations suffices in scenarios such as stock market data~\cite{memon2019structural}.
For future work, we plan to extend our problem definition to  unaligned time-series, develop a pairwise correlation  pruning algorithm based on a threshold, and  consider further optimization of the parallel \name by fine-tuning.} 



\newpage

\bibliographystyle{ACM-Reference-Format}
\bibliography{ref}

\end{document}